%% file: main.tex
\documentclass[conference]{IEEEtran}

\usepackage{booktabs} 
\usepackage{graphicx}
\usepackage{comment}
\usepackage{amsmath,amssymb,amsfonts}
\usepackage{wrapfig,paralist}
\usepackage{lscape}
\usepackage{rotating}
\usepackage{epstopdf}
\usepackage{color,balance}
\usepackage{amsthm}
\usepackage[boxruled,linesnumbered]{algorithm2e}
\usepackage[hidelinks]{hyperref}
\usepackage[capitalise, noabbrev]{cleveref}
\usepackage{savesym}
\savesymbol{state}
\usepackage{tengwarscript}
\usepackage{xcolor}
\usepackage{url}

\usepackage[noend]{algpseudocode}
\usepackage{array}
\let\labelindent\relax
\usepackage{enumitem}
\usepackage{xspace}
\usepackage{multirow}
\usepackage{subcaption}
\usepackage{pdflscape}
\usepackage{mathtools,balance}
\usepackage{booktabs}
\usepackage{float}
\usepackage{newfloat}
\usepackage{tikz}
\usepackage[sort]{cite}

\input{macros}
\input{math_commands}

\usepackage{colortbl}
\newcolumntype{H}{>{\setbox0=\hbox\bgroup}c<{\egroup}@{}}

\SetCommentSty{mycommfont}

\begin{document}

\author{\IEEEauthorblockN{Mohsen Minaei\authorrefmark{2}\authorrefmark{1},
S Chandra Mouli\IEEEauthorrefmark{2}\authorrefmark{1},
Mainack Mondal\IEEEauthorrefmark{4}, 
Bruno Ribeiro\IEEEauthorrefmark{2},
Aniket Kate\IEEEauthorrefmark{2}}
\IEEEauthorblockA{\IEEEauthorrefmark{2} Purdue University,\quad Email: \{mohsen, chandr, ribeirob, aniket\}@purdue.edu}
\IEEEauthorblockA{\IEEEauthorrefmark{4} IIT Kharagpur,\quad Email: mainack@cse.iitkgp.ac.in}
}

 \title{Deceptive Deletions for Protecting Withdrawn Posts on Social Platforms\vspace{-2mm}} 
\maketitle
\input{abstract}

\input{intro}

\input{related_work}

\input{system_model}

\input{formal_analysis}
\input{twitter_experiment}

\input{results}

\input{design_rationale}

\input{conclusion}

\Urlmuskip=0mu plus 1mu\relax
{\footnotesize
	\bibliographystyle{IEEEtranS}
	\bibliography{main}
}
\balance
 \appendix
 \input{appendix}
\input{amt_details}

\end{document}

%% file: macros.tex
\newtheorem{definition}{Definition}
\newtheorem{proposition}{Proposition}

\newtheorem*{proposition*}{Proposition}

\newcommand{\todo}[1]{\iftrue\textcolor{red}{{TODO: }{#1}}\fi}

\newcommand{\scm}[1]{\iftrue{\color{orange}SCM:#1}\fi}

\renewcommand\paragraph[1]{\vspace{0.3mm}\noindent \textbf{#1.\ }}

\newcommand{\system}{$\text{D}^{\text{2}}$\xspace}
\newcommand{\sysname}{Deceptive Deletions\xspace}
\newcommand{\game}{deceptive learning game\xspace}
\newcommand{\fullgame}{Deceptive Learning Game\xspace}

\newcommand{\op}{volunteered posts\xspace}

\newcommand{\damaging}{damaging\xspace}
\newcommand{\Damaging}{Damaging\xspace}
\newcommand{\nondamaging}{non-damaging\xspace}
\newcommand{\Nondamaging}{Non-damaging\xspace}

\newcommand{\dt}{\ensuremath{t}\xspace}
\newcommand{\damset}{\ensuremath{+}}

\newcommand{\decoyset}{\ensuremath{*}}
\newcommand{\offeredset}{\text{v}}
\newcommand{\delset}{\ensuremath{\delta}}

\newcommand{\sample}[1]{\ensuremath{\overset{#1}{\sim}}}

\newcommand{\cD}{\mathbb{D}}

\newcommand{\adv}{a}
\newcommand{\con}{g}

\newcommand{\budgetSA}{B_\text{static}}
\newcommand{\budgetAA}{B_\text{adapt}}
\newcommand{\budgetC}{B_\con}

\newcommand{\adaptive}{adaptive\xspace}
\newcommand{\nonadaptive}{static\xspace}

\newcommand{\Confuser}{Challenger\xspace}
\newcommand{\confuser}{challenger\xspace}
\newcommand{\confusers}{challengers\xspace}
\newcommand{\adversary}{adversary\xspace}
\newcommand{\epoch}{time interval\xspace}
\newcommand{\epochs}{time intervals\xspace}

\newcommand{\decoy}{decoy\xspace}

\newcommand{\subtest}{\text{test}}
\newcommand{\subtrain}{\text{train}}

\newcommand{\setDecoyPostsNotime}{\cD^\offeredset}
\newcommand{\setDecoyPostsTrainNotime}{\cD^{\offeredset, \subtrain}}
\newcommand{\setDecoyPostsTestNotime}{\cD^{\offeredset, \subtest}}

\newcommand{\setDecoyPosts}{\cD_{\leq t}^\offeredset}
\newcommand{\setDecoyPostsTrain}{\cD_{\leq t}^{\offeredset, \subtrain}}

\newcommand{\Nofferedlt}{\ensuremath{N^\offeredset}}

\newcommand{\NofferedltTrain}{\ensuremath{\budgetC}}

\newcommand*\circled[1]{\tikz[baseline=(char.base)]{
            \node[shape=circle,fill,inner sep=0.8pt] 
            (char) {\textcolor{white}{#1}};}}

\newcommand\blfootnote[1]{
	\begingroup
	\renewcommand\thefootnote{}\footnote{#1}
	\addtocounter{footnote}{-1}
	\endgroup
}

%% file: math_commands.tex
\usepackage{amsmath,amsfonts,bm}

\def\eqref#1{equation~\ref{#1}}

\def\1{\bm{1}}

\def\vtheta{{{\theta}}}
\def\vphi{{{\phi}}}

\def\vw{{\bm{w}}}

\def\evw{{w}}

\DeclareMathAlphabet{\mathsfit}{\encodingdefault}{\sfdefault}{m}{sl}
\SetMathAlphabet{\mathsfit}{bold}{\encodingdefault}{\sfdefault}{bx}{n}

\def\sA{{\mathbb{A}}}

\def\sG{{\mathbb{G}}}

\def\sR{{\mathbb{R}}}

\def\sX{{\mathbb{X}}}

\newcommand{\Ls}{\mathcal{L}}

\DeclareMathOperator*{\argmax}{arg\,max}
\DeclareMathOperator*{\argmin}{arg\,min}

%% file: abstract.tex
\begin{abstract}
Over-sharing poorly-worded thoughts and personal information is prevalent on online social platforms. 
In many of these cases, users regret posting such content.
To retrospectively rectify these errors in users' sharing decisions, most platforms offer (deletion) mechanisms to withdraw the content, and social media users often utilize them. 
Ironically and perhaps unfortunately, these deletions make users more susceptible to privacy violations by malicious actors who specifically hunt post deletions at large scale. The reason for such hunting is simple: deleting a post acts as a powerful signal that the post might be damaging to its owner.
Today, multiple archival services 
are already scanning social media for these deleted posts.
Moreover, as we demonstrate in this work, powerful machine learning models can detect damaging deletions at scale.

Towards restraining such a global adversary against users' right to be forgotten, 
we introduce \textit{Deceptive Deletion}, 
a decoy mechanism that minimizes the adversarial advantage. Our mechanism injects decoy deletions, hence creating a two-player minmax game between an {\em \adversary} that seeks to classify damaging content among the deleted posts and a {\em \confuser} that employs decoy deletions to masquerade real damaging deletions. 
We formalize the Deceptive Game between the two players, determine conditions under which either the adversary or the challenger provably wins the game, and discuss the scenarios in-between these two extremes.
We apply the \textit{Deceptive Deletion} mechanism to a real-world task on Twitter: hiding damaging tweet deletions. We show that a powerful global adversary can be beaten by a powerful \confuser, raising the bar significantly and giving a glimmer of hope in the ability to be \textit{really} forgotten on social platforms.
\end{abstract}

%% file: intro.tex
\section{Introduction}\label{sec:intro}
\blfootnote{\noindent*Both authors contributed equally and are considered co-first authors.}
Every day, millions of users share billions of (often personal) posts on online social media platforms like Facebook and Twitter.
This information is routinely archived and analyzed by multiple third parties ranging from individuals to state-level actors ~\cite{andre-2012-TwitterContentEval,huberman-2009-TwitterMicroscope,java-2007-whyweTwitter,ottowaelection2010-2014-raynauld,politicsteaparty-2013-reynolds,choudhury-2013-depressionTwitter,tsugawa-2015-depressionTwitter,pedersen-2015-ptsdTwitter}. 
Although the majority of these social media posts are benign, users also routinely post regrettable content on social media ~\cite{sleeper-2013-regretTweetsCHI, zhou-2016-regrettableDelTweet, bauer-2013-postanachronism} that they later wish to retract.
Subsequently, most social platforms provide user-initiated deletion mechanisms that allow users to rectify their sharing decisions and delete past posts. 
 Not surprisingly, users take advantage of these deletion mechanisms enthusiastically---Mondal et al.~\cite{mondal-2016-longitudinal-exposure}  showed that nearly one-third of six-year-old Twitter-posts were deleted. In another work, Tinati et al.~\cite{tinati2017instacan} showed that this number is much higher in Instagram, where almost half of the pictures posted within a six month period had been removed.

Ironically, current user-initiated deletion mechanisms may have an unintended effect: third-party archival services can identify deleted posts and infer that deleted posts might contain \textit{damaging} content from the post creator's point of view (i.e., having an adverse effect on the personal/professional life of the content creator). 
In other words, deletion might inadvertently make it easier to identify damaging content.
Indeed, today it is possible to detect deletions at scale: Twitter, for one, advertises user deletions in their streaming API
\footnote{Twitter provides a random sample of the publicly posted Twitter data in real time to the third parties via streaming API.} 
via deletion notifications~\cite{Twitterapi,Twitter-deletion-notification} so that third-party developers can remove these posts from their database.
Similarly, Pushshift~\cite{baumgartner2020pushshift,pushshift} is an archival system for all the contents on Reddit and Removeddit~\cite{removeddit} uses this archive to publicize all the deleted posts and comments on Reddit.
A malicious data-collector can simply leverage these notifications to flag deleted posts as possibly \textit{damaging} and further use 
them 
against the users~\cite{ed_shereen_buzzfeed,SNL,xue2016right}. 
Importantly, the hand-picked politicians and celebrities are {\em not} the only parties at the receiving end of these attacks. 
We find that the malicious data-collector can develop learning models to automate the process and perform an non-targeted (or global) attack at a large-scale;
e.g., Fallait Pas Supprimer~\cite{fallait} (i.e., ``Should Not Delete'' in English) is a Twitter account that collects and publishes the deleted tweets of not only the French politicians and celebrities but also noncelebrity French users with less than a thousand followers.

Asking the users not to post regrettable content on social platforms in the first place may seem like a good first step. 
However, users cannot accurately predict what content would be damaging to them in the future (e.g., after a breakup or before applying to a job). 
Zhou et al.~\cite{zhou-2016-regrettableDelTweet} and Wang et al.~\cite{wang2019donttweetthis} 
propose multiple types of classifiers (Naive Bayes, SVM, Decision Trees, and Neural Networks) to detect regrettable posts using users' history and
to proactively advise users even before the publication of posts. However, this proactive approach cannot prevent users from publishing future-regrettable posts. 
It is inevitable to focus on {\em reactive} mechanisms to assist users with protecting their post deletions.
 
Recently Minaei et al.~\cite{lethePets2019} proposed an intermittent withdrawal mechanism to tackle this challenge of hiding user-initiated deletions. 
They offer a deniability guarantee for user-initiated deletions in the form of an availability-privacy trade-off and ensure that when a post is deleted, the adversary cannot be immediately certain if it was actually deleted or temporarily made unavailable by the platform. 
Their trade-off could be useful for future social and archival platforms; however, in current commercial social media platforms like Twitter, sacrificing even a small fraction of availability for \textit{all the posts} is undesirable.

To this end, our research question is straightforward, yet highly relevant---\textit{can we enhance the privacy of the deleted and possibly \textit{damaging} posts at scale without excessively affecting the functionality of the platform?}

\paragraph{Contributions}
We make the following contributions.

First, we demonstrate the impact of deletion detection attacks by performing a proof-of-concept attack on real-world social media posts to identify \damaging content.
Specifically, we use a crowdsourced labeled corpus of (non)damaging deleted posts from Twitter (more than $4,000$ tweets) to train an adversary (a classifier). 
We demonstrate that our adversary is capable of detecting \damaging posts with high probability. More precisely, our adversary can increase its F-score by 27 percentage points (56\% increase) compared to a baseline adversary which uses random guessing to detect \damaging posts. Thus, it is indeed feasible for the adversary to use automated methods for detecting \damaging posts at a large scale (e.g., when focusing only on deleted posts).
In fact, we expect systems such as Fallait Pas Supprimer~\cite{fallait} to employ analogous learning techniques soon to improve their detection.  

Second, we identify that there are already deletion services which enable users to delete their content in bulk (e.g., ``twitWipe''~\cite{twitWipe} and ``tweetDelete''~\cite{tweetDelete} for Twitter, ``Social Book Post Manager''~\cite{facebookcleaner2} for Facebook, ``Cleaner for IG''~\cite{instacleaner} for Instagram,  ``Nuke Reddit History''~\cite{redditremoverextension}, and multiple bots on RequestABot subreddit for Reddit). However, these bulk deletions provide a clear signal to an adversary that the user is trying to hide \damaging content via deletion.
To that end, we introduce a novel deletion mechanism, Deceptive Deletions, 
that raises the bar for the adversary in identifying \damaging content.
Given a set of damaging posts (posts that adversary can leverage to blackmail the user) that users want to delete, 
the Deceptive Deletion system (also known as a \confuser)
\textit{selects} $k$ additional posts for each damaging post and deletes them along with the damaging posts.
The system-selected posts, henceforth called the \textit{decoy posts}, are taken from a pool of non-damaging non-deleted posts provided by volunteers. 
Since a global adversary can only observe all of these deletions together, his goal is to distinguish deleted damaging posts from the deleted (non-damaging) decoy posts. 
Intuitively, Deceptive Deletion is more effective if the selected decoy posts are similar to the damaging posts. 
These two opposite goals create a minmax game between the adversary and the \confuser that we further analyze.

Third, we introduce the \fullgame, which formally describes the minmax game between the adversary and the \confuser.
We start by considering a \nonadaptive adversary that tunes the parameters of its system (e.g., classifier for determining the \damaging posts) up until a certain point in time.
However, powerful adversaries are \adaptive and continually tune their models as they obtain more deletions including the decoy deletions made by the \confuser. 
Therefore, in the second phase, we consider an \adaptive adversary and describe the optimization problem of the \adaptive adversary and \confuser as a minmax game.

We identify conditions under which either only the \adaptive adversary or only the \confuser provably wins the minmax game and discuss the scenarios in-between these two extremes.
To the best of our knowledge, this is the first attempt to develop a computational model for quantitative assessment of the damaging deletions in the presence of both \nonadaptive and \adaptive adversaries.

Finally, we empirically demonstrate that with access to a set of non-damaging volunteered posts, we can leverage 
Deceptive Deletions
to hide damaging deletions against both \nonadaptive and \adaptive adversary effectively.
We use real-world Twitter data to demonstrate the effectiveness of the \confuser. 
Specifically, we show that even when we consider only two decoy posts per damaging deletion, the adversarial performance (F-score) drops to 42\% from 75\% in the absence of any privacy-preserving deletion mechanism.

%% file: related_work.tex
\section{Background and Related Work} \label{sec:related_work}

\subsection{Exisiting Content Deletion Mechanisms to Provide Privacy}
\noindent Today, most archival and social media websites (e.g., Twitter, Facebook) 
enable users to delete their content. 
Recent studies ~\cite{mondal-2016-longitudinal-exposure,almuhimedi-2013-deleteTweets} show that a significant number of users deleted content---35\% of Twitter posts are deleted within six years of posting them. This user-initiated deletion is also related to the ``Right to be Forgotten" ~\cite{xue2016right,weber-2011-rightToForget}. However, this user-initiated content deletion suffered from the \textit{Streisand effect} -- attempting to hide some information has the unintended consequence of gaining more attention~\cite{xue2016right}. Consequently, there is a need to provide \textit{deletion privacy} to users. 

In addition to user-initiated deletions, there exist some premeditated withdrawal mechanisms where all historical content is eventually deleted automatically to provide deletion privacy. These mechanisms can be broadly classified into two categories. First, in \textit{age-based withdrawal}, platforms like Snapchat~\cite{snapchat}  and Dust~\cite{dust} and systems like Vanish~\cite{geambasu-2009-vanish, geambasu-2011-vanishImprove} and EphPub~\cite{ephpub-2011-Castelluccia} automatically withdraw a piece of content after a preset time. Second, to make premeditated withdrawal more usable, 
Mondal et al.~\cite{mondal-2016-longitudinal-exposure} 
proposed \textit{inactivity-based withdrawal}, where  posts will be withdrawn only if they become inactive, 
i.e., there is no interaction with the post for a specified time period (e.g., no more views by other users).

However, even the premeditated withdrawals are not free from problems of their own. First, all the posts will eventually get deleted,
removing all archival history from the platform. 
Second, if posts are deleted before the preset time or in-spite of high interaction,
the adversary can be certain that the deletion was user-intended, violating deletion privacy.

\noindent Minaei et al.~\cite{lethePets2019} attempted to enable users to delete their content while neither attracting attention to deleted content nor deleting full historical archives. They presented a new intermittent withdrawal mechanism for all non-deleted posts, which provides a trade-off between availability and deletion privacy. In a nutshell, their system ensures that if an adversary found that a post if not available, then the adversary cannot be certain if the post is user-deleted or simply taken down by platform temporarily. Although this mechanism is useful for large internet archives, in platforms such as Facebook and Twitter, where content availability is crucial to the users and platform,
a privacy-availability trade-off might not be feasible.
Furthermore, the intermittent withdrawal mechanism does not consider the adversary's background knowledge about other deleted posts. Our work aims to bridge this gap and provide a novel learning-based mechanism which considers an adaptive adversary who aims to uncover tweet deletion. However, our mechanisms is not without precedent, and it is inspired by earlier work of obfuscation by noise injection.

Tianti et al.\ \cite{tinati2017instacan} offer intuitions for predicting posts deletions on Instagram with the goal of managing the storage of posts on the servers: Once a post is archived, it becomes computationally expensive to erase it; thus, predicting deletions can help in reducing the overheads of being compliant with the ``right to be forgotten'' regulations. These predictions in the non-adversarial setting, however, does not apply to our minmax game between the adversary and the challenger.

Recently Garg et al.~\cite{garg2020formalizing} formalize the right to be forgotten using platforms as a cryptographic game.
While being interesting, their definitions and suggested tools
such as history-independent data structures are not applicable
to our setting where the adversary has continuous and complete access to the collected data.

\subsection{Obfuscation using Noise Injection}

\noindent There has been a line of work in the area of (non-cryptographic) private information retrieval~\cite{howe2009trackmenot,murugesan2009providing,domingo2009h,peddinti2010privacy} that obfuscates the users' interest using dummy queries as noise to avoid user profiling.
Howe et al. proposed TrackMeNot~\cite{howe2009trackmenot, trackMeNotSite}, which issues randomized search queries to popular search engines to prevent the search engines in building any practical profile of the users based on their actual queries.  
GooPIR~\cite{domingo2009h} is a similar work that uses a Thesaurus to obtain the keywords to constructs $k - 1$ other
queries (dummy ones) and submit all $k$ queries at the same time.
This way, timing attacks by the search engines are eliminated. 
However, it only addresses single keyword searches; 
these schemes do not address full-sentence searches. 
Murugesan et al. propose ``Plausibly Deniable
Search" (PDS)~\cite{murugesan2009providing} that analogous to GooPIR generates $k-1$ dummy queries using latent
semantic indexing based approach. In their mechanisms, 
each real query is converted into a canonical query which protects against deanonymization attacks based on typos and grammar mistakes.
We note that all of the systems mentioned so far 
consider hiding each query separately. However,
a determined adversary may be able to find a user's interests by observing a sequence of such obfuscated queries.
Multiple works have investigated such weaknesses~\cite{peddinti2010privacy,peddinti2011effectiveness,balsa2012ob}.

Some relatively new techniques further try to overcome these shortcomings by \textit{smartly} generating the $k-1$ queries. 
For example, Petit et al. proposed PEAS~\cite{petit2015peas}, where they provide a combination of unlinkability and indistinguishability. 
However, apart from introducing an overhead for encrypting the user queries, their method also requires two proxy servers that are non-colluding, hence weakening the adversarial model.
K-subscription~\cite{papadopoulos2013k} is yet another work that proposes an obfuscation based approach that enables the user to follow privacy-sensitive channels in Twitter by requiring the users to follow $k-1$ other channels to hide the user interests from the microblogging service. 
However, the K-subscription has a negative social impact for the user as the user's social connections will see the user following these dummy channels. These shortcomings, both social and technical, motivated our particular design decision for Deceptive Deletions.

%% file: system_model.tex
\section{System Model and Overview} \label{sec:sys_threat_model}

\begin{figure*}[!ht]
	\centering
	\includegraphics[width = \textwidth]{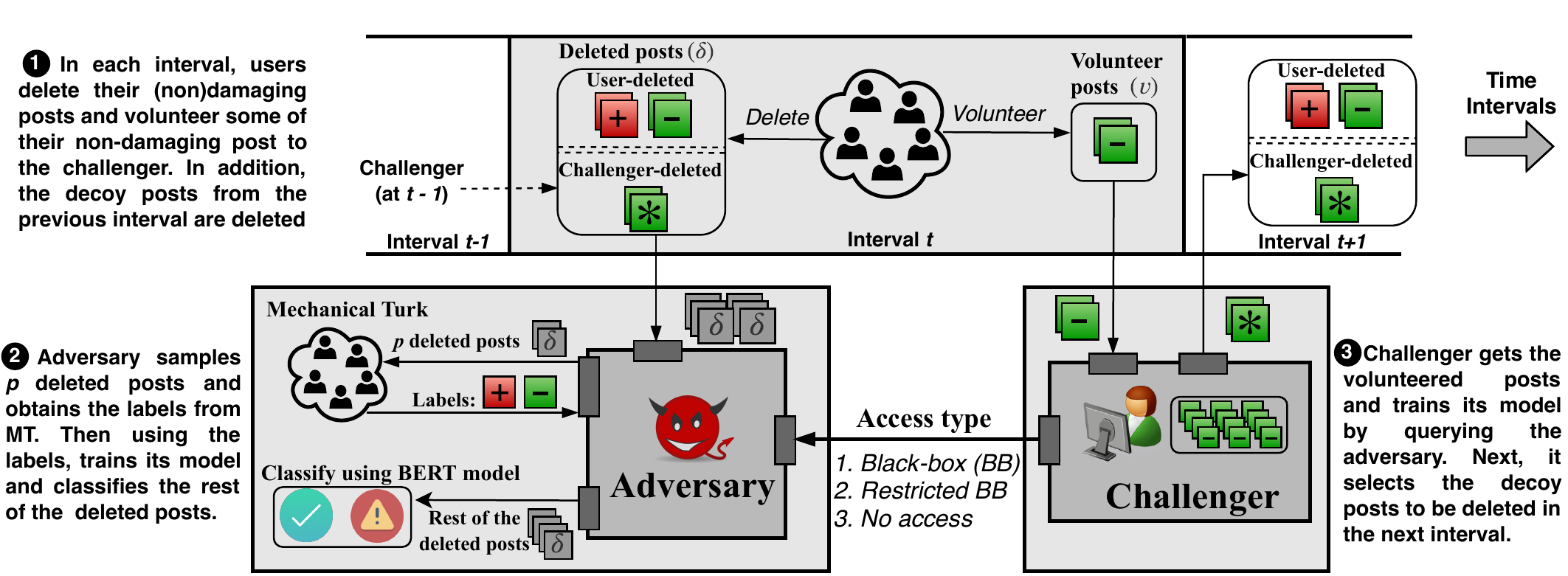}
	\caption{{Overview of Deceptive Deletions. 
	In each interval, the deletions are shown by gray squares with `$\delta$'. The deleted posts could be of three types: users' damaging deletions shown by red squares with `$+$', users' non-damaging deletions shown by green squares with `$-$' and challengers' decoys posts shown by green squares with `$*$'. Further, we denote the volunteer posts offered to the challenger during each interval by green squares with `$-$' to indicate that they are non-damaging.}}
	\label{fig:system_overview}
\end{figure*}

\subsection{System}
We consider a data-sharing \textit{platform} (e.g., Twitter or Facebook) as the public bulletin board where individuals can upload and view content.
\textit{Users} are the post owners that are able to publish/delete their posts, and view posts from other users. In this work, we consider discrete time intervals in which the users upload and delete posts (\cref{fig:system_overview} \circled{1}). A time interval could be as small as a minute or even a week, depending on the platform. 
We define two types of posts.

\begin{itemize}[leftmargin = *]
	\item \textbf{User-deleted posts} A user could delete a post for two primary reasons~\cite{lethePets2019, almuhimedi-2013-deleteTweets, mondal-2016-longitudinal-exposure}:
	    	\begin{itemize}
    		\item \textbf{\Damaging posts:} the post contained \textit{damaging} content to the user's personal or professional life, or
    		\item \textbf{\Nondamaging posts:} the post was out-dated, contained spelling mistakes, etc. 
    	\end{itemize}
		
    	\noindent \textit{An adversary's goal is to find the damaging posts among all the deleted ones that could be used to blackmail the corresponding owners of the post.}
    \vspace{1mm}	
	\item\textbf{Volunteered posts} We consider a subset of non-deleted posts that users \textit{willingly} offer to be deleted to protect the privacy of other users whenever needed. These volunteered posts are non-damaging and cannot be used by the adversary to blackmail the user of the post. We discuss the challenges of obtaining volunteered posts in \cref{sec:discussion}.
        
    	\noindent \textit{A challenger's goal is to select a subset of volunteered posts (i.e., \nondamaging) and delete them such that the aforementioned adversary is unable to distinguish between the \damaging and the \nondamaging post deletions. We denote the posts selected by the challenger as \textbf{decoy posts}.}
\end{itemize}

\paragraph{Notation} We use a subscript $\dt$ to denote the time interval and superscripts $\delset, \damset,
\offeredset, \decoyset$ to denote the post type. In particular, $\cD_\dt$ is all the uploaded and deleted posts in \epoch $\dt$. Then we denote all the deleted posts (user- and challenger-deleted) in that interval as $\cD_\dt^\delset$, 
the \textbf{\damaging posts} as $\cD_\dt^\damset$, and
\textbf{volunteered posts} by $\cD_\dt^\offeredset$. The \textbf{decoy posts} that a challenger selects for deletion to fool the adversary is denoted by $\sG_\dt^\decoyset$. Note that $\sG_\dt^\decoyset \subseteq \cD_\dt^\offeredset \subseteq \cD_\dt {\setminus} \cD_\dt^\delset$.

\subsection{Adversary's Actions and Assumptions} \label{sec:threatmodel}


\paragraph{Task}
At a given \epoch, the task of the adversary is to correctly label all the deleted posts as being damaging to the post-owner or not. 
We {\em do not} focus on local attackers (or stalkers) targeting {individuals} or small groups of users.\footnote{Such stalkers can easily label their posts manually, and protecting against such an attack is extremely hard if not impossible.
For example, consider the case that a stalker continuously takes snapshots of its targeted user profile with the goal of identifying the user's deletions. With its background/auxiliary information about the user (i.e., knowing what contents are considered sensitive to the target), the stalker can effectively identify the damaging deletions. We claim that, in the current full-information model, protection against such a local adversary is impossible.}
Our global adversary instead seeks \damaging deletions on a large scale, rummaging through all the deleted posts to find as many \damaging ones as possible. 
Fallait Pas Supprimer~\cite{fallait} (from \cref{sec:intro}) is a real-world example of the global adversary.

\paragraph{Data access}
At any given \epoch, we assume that the adversary is able to obtain all the deleted posts by comparing different archived snapshots of the platform. 
{Although this strong data assumption benefits the adversary tremendously, we show in~\cref{sec:results} that Deceptive Deletions can protect the users' damaging deletions.}

\paragraph{Labels}
Our global, non-stalker adversary is not able to obtain the \textit{true label} (\damaging or \nondamaging) of the post from the user.
Instead, the adversary uses a crowdsourcing service like Mechanical Turk (MTurk)~\cite{amt} to obtain a proxy for these true labels. Although the labels obtained from the Mechanical Turkers (MTurkers) reflect societal values and not the user's intention, following previous work~\cite{wang2019donttweetthis}, we assume they closely match the \textit{true labels} in our experiments. 
This is reasonable as the adversary can expend a significant amount of effort and money to obtain these \textit{true labels}, at least for a small set of posts, that will ultimately be used to train a machine learning model. 

\paragraph{Budget}
Since there is a cost associated with acquiring label for each deleted post from the MTurkers, the aim of the adversary is to \textit{learn to detect the damaging deletions} under a budget constraint. We consider two types of budget constraints:

\begin{itemize}[leftmargin=*]
	\item \textbf{limited budget} where the adversary can only obtain the labels for a fixed number of posts $\budgetSA$, and

	\item \textbf{fixed recurring budget} where the adversary obtains the labels for a fixed number of posts $\budgetAA$ \textit{in each interval}.
\end{itemize}

The adversary with a limited budget is called the \textbf{\nonadaptive adversary} since it does not train after exhausting its budget. On the other hand, the adversary with a fixed recurring budget keeps adapting to the new deletions in each time interval, and hence is dubbed the \textbf{\adaptive adversary}.

\paragraph{Player actions}
At every \epoch $\dt$, the adversary obtains a set of posts $\sA_\dt^\delta$ for training by sampling part of the deleted posts, say $p$, from $\cD_\dt^\delta$,
an operation denoted by $\sA_\dt^\delta~\sample{p}~\cD_\dt^\delta$.
The adversary uses MTurk to label the sampled dataset $\sA_\dt^\delta$. After training, the task of the adversary is to classify the rest of the deleted posts of that \epoch.
Additionally, as the adversary gets better over time, it also \textit{relabels} all the posts deleted from the past intervals. The test set for the adversary is all the deleted posts from current and previous time intervals that were not used for training; i.e.,  $\bigcup_{\dt' \leq \dt} (\cD_{\dt'}^\delta \setminus \sA_{\dt'}^\delta)$. \cref{fig:system_overview}~\circled{2} shows the adversary's actions.

Note that although an \adaptive adversary can sample $p = \budgetAA$ deleted posts at \textit{every} \epoch and use MTurkers to label them, a \nonadaptive adversary can only obtain the labels until it runs out of the limited budget (after $\tau = \budgetSA/p$ \epochs). After this period, a \nonadaptive adversary does not train itself with new deleted posts. 

\paragraph{Performance metrics}
The adversary wishes to increase \textit{precision} and \textit{recall} for the classification of deleted posts into \damaging and \nondamaging sets. At every \epoch $\dt$, we report adversary's F-score\footnote{F-score = $2\cdot precision\cdot recall/(precision + recall)$} over the test set described above: deleted posts of all the past intervals, i.e., $\bigcup_{\dt' \leq \dt} (\cD_{\dt'}^\delta \setminus \sA_{\dt'}^\delta)$.

\subsection{Challenger's Actions and Assumptions} \label{subsec:threat_challenger}

\paragraph{Task} In the presence of an adversary as described above, the task of a challenger is to obtain \op (i.e. non-damaging and non-deleted posts) from users, select a subset of these posts and delete them in order to fool the adversary into misclassifying these \textit{challenger-deleted} posts as \damaging.
The challenger is honest, does not collude with the adversary, and works with the users (data owners) to protect their damaging deletions.
Other than the platforms themselves, third party services such as ``tweetDelete''~\cite{tweetDelete} 
can take the role of the challenger as well.
In \cref{sec:gan}, we discuss the flaws in a possible alternate approach where the challenger is allowed to \textit{generate} tweets rather than \textit{select} from pool of volunteered posts.

\paragraph{Data access}
{The challenger can be implemented by the platform or a third-party deletion service~\cite{tweetEraser,tweetDelete,tweetDeleter}, that has access to the posts of the users.}
Additionally, we assume that there are users over the platform who volunteer a subset of their non-damaging posts to be deleted anytime (or within a time frame) by the challenger, possibly, in return for privacy benefits for their (and other users') damaging deletions.

\paragraph{Labels}
{The challenger is implemented as part of the platform (or a third-party service permitted by the user). Thus, unlike the adversary that obtains proxy labels from crowdsourcing platforms, it has access to the \textit{true labels}--- \damaging or \nondamaging, from the owner of the post. This is easily implemented: before deleting a post, the user can specify whether the post is damaging (and needs protection). 
}

\paragraph{Access to the adversary}
The challenger not only knows the presence of a global adversary trying to classify the deleted posts into damaging and non-damaging posts but also can observe its behaviour.\footnote{Fallait Pas Supprimer~\cite{fallait} 
posts all its output on Twitter itself}. As a result, we consider three types of accesses to the adversary:

\vspace{-2mm}
\begin{itemize}[leftmargin=*]
	
	\item \textbf{no access} where the challenger has no information about the adversary.

	\item \textbf{monitored black-box access with a recurring query budget of $\budgetC$} where the challenger can obtain the adversary's classification probability for a limited number of posts $\budgetC$ every time interval, but the access is \textit{monitored}, i.e., the adversary can take note of every post queried and treat them separately. 
	
	\item \textbf{black-box access} where the challenger can obtain the adversary's classification probabilities for any post.
	
\end{itemize}
Here, {\em no access} is the weakest assumption that defines
the lower-bounds for our challenger's success. Nevertheless, we expect
the challenger to have some access to the adversary's classification. 
An \textit{unrestricted} black-box access 
serves as an upper bound for the challenger assuming that it can train a precise surrogate model of the adversary's classifier using its own training data. While employing such a surrogate model is common practice in the literature~\cite{kurakin2016adversarial,Papernot:2017:PBA:3052973.3053009}, it can be hard to obtain in real world without knowing the adversary's exact architecture and training data.
Our monitored black-box assumption with a recurring query budget (henceforth, interchangeably called the restricted black-box access) balances practicality of the access versus the feasibility of defending against an adversary with that access. 
In \cref{sec:model}, we introduce three challengers (oracle, \system and random) corresponding to the three types of accesses.

\paragraph{Player actions}
{At every \epoch $\dt$, the challenger receives new volunteer posts from the users and adds them to a set that stores the volunteered posts collected up until this point. 
Next, based on the type of access, it obtains the adversary's classification probabilities for some number of volunteer posts (the number is dependent on the access which we detail in~\cref{sec:model}). 
Finally, it selects \textit{decoy posts}, a subset of the volunteered posts collected up until this point
and deletes these posts in interval $\dt{+}1$ (hence the adversary sees these \textit{challenger-deleted} posts in interval $\dt{+}1$ as part of the deleted set $\cD_{\dt{+}1}^\delta$). 
\cref{fig:system_overview}~\circled{3} shows the challenger's actions.
}

\paragraph{Performance metrics}
The challenger, in direct contrast to the adversary, wishes to \textit{decrease} adversary's precision and recall for the classification of deleted posts. Adversary's precision will decrease if it classifies the injected \decoy posts as \damaging (increased false-positives). On the other hand, adversary's recall will decrease if it learns to be conservative in order to ignore the \decoy posts (increased false-negatives).

%% file: formal_analysis.tex
\section{The Deceptive Learning Game}\label{sec:model}

The \game is a two-player zero-sum non-cooperative game over time intervals $t=1, 2, ... $ (units) between an adversary who wishes to \textit{find} users' damaging deletions, and a challenger who wishes to \textit{hide} the said damaging deletions. The challenger achieves this by deleting volunteers' non-damaging posts as decoys. While the adversary's goal is to maximize its precision/recall scores on the classification task, the challenger's goal is to minimize them. 

We denote each post by $(x, y)$, where $x \in \sX$ represents the features of the post (i.e., text, comments, etc.) and $y \in \{0, 1\}$ denotes its true label such that $y=1$ if the post is damaging and $y=0$ if it is non-damaging. 
In the following subsections, we describe the actions of each player in the time interval $t$.

\subsection{Adversary} \label{subsec:adversary}

We denote the adversary's classifier at the beginning of interval $\dt$ by $\adv(~\cdot~; \vtheta_{t-1}) : \sX \rightarrow [0, 1]$ parameterized by $\vtheta_{t-1}$ such that $\adv(x; \vtheta_{t-1}) \coloneqq P(\hat{y}=1 ~|~ x; \vtheta_{t-1})$ is the predicted probability of the post $x$ being damaging. The adversary collects all the deletions that happen in this interval (i.e., $\cD_\dt^\delta$) and samples $p$ posts, denoted by $\sA_\dt^\delta$. The adversary then uses MTurk to obtain a proxy for the true labels 
of these $p$ posts. 

The adversary uses this labeled training data in the following optimization problem to update its parameters, 
\vspace{-1mm}
\begin{align} \label{eq:adversary}
\vtheta_{t} = \argmin_{\vtheta} \Ls_\text{NLL}(\vtheta; \sA_\dt^\delta) \:,
\end{align}
\vspace{-2mm}

\noindent where $\Ls_\text{NLL}$ is the standard negative log-likelihood loss for the classification task, given by, 
\vspace{-1mm}
\small
\begin{align*}
\Ls_\text{NLL}&(\vtheta; \sA_\dt^\delta) = \sum_{(x, y) \in \sA_\dt^\delta} -y \log\left(\adv(x; \vtheta)\right) - (1 - y) \log\left(1 - \adv(x; \vtheta)\right) \:. 
\end{align*}
\normalsize
After training, the adversary uses the {trained} model $a(~\cdot~; \vtheta_{t})$ to predict the labels of the rest of the deleted posts of time interval~$t$, i.e., $\cD_\dt^\delta \setminus \sA_\dt^\delta$ along with all the deleted posts that it had already predicted in the past. This way the adversary hopes to capture damaging posts that were missed earlier. Hence, we report the adversary's performance on all the past deletions (not including the training data): $\bigcup_{\dt' \leq \dt} (\cD_{\dt'}^\delta \setminus \sA_{\dt'}^\delta)$. 

\vspace{2mm}
\paragraph{Static vs Adaptive Adversary}
Since the static adversary has a limited budget, first it chooses the number of time intervals for training, say $\tau$, and accordingly samples $p = \budgetSA / \tau$ posts for querying MTurk to obtain labels. 

The adaptive adversary has a fixed recurring budget of $\budgetAA$ and hence, can sample $p=\budgetAA$ posts every interval. This allows the adaptive adversary to train itself with new training data (of size $\budgetAA$) every interval indefinitely. 
\cref{alg:adversary} depicts adversary's actions within a time interval (subscript $t$ removed for clarity).

\begin{algorithm}[t]
	\SetKwInOut{Input}{input}\SetKwInOut{Output}{output}
	\SetAlgoLined
	
	\Input{ $\cD^{\delset}$\tcc*[r]{Deleted posts in this interval}}

	Sample $p$ posts $\sA^\delta~\sample{p}~\cD^\delta$\;
	Query MTurk and obtain labels for $\sA^\delta$ \;
	Obtain optimal parameters $\vtheta^* $  by solving~\cref{eq:adversary} \;
	%
	\Return $\adv(~\cdot~;\vtheta^*)$
	\caption{Adversary \label{alg:adversary} }
\end{algorithm}

\subsection{Challenger} \label{subsec:challenger}

In the presence of such an adversary, the challenger's goal is to collect volunteered posts (non-damaging) from users and selectively delete these posts in order to confuse the adversary. 

As described before, $\cD_t^\offeredset$ is the set of posts volunteered by users in the time interval $t$. Let $\sG_{\leq t}^\decoyset$ be the set of decoy posts deleted by the challenger in the current and past intervals. 
At the end of interval $t$, the challenger collects all the volunteered posts from the current and past intervals (except the posts that it has already used as decoys).  
The \textit{available} set of volunteered posts is denoted by $\cD_{\leq t}^\offeredset \equiv  (\bigcup_{\dt'\leq\dt} \cD_{\dt'}^\offeredset) \setminus (\bigcup_{\dt'\leq\dt} \sG_{\dt'}^\decoyset$).
Note that $(x, y) \in \cD_{\leq t}^\offeredset {\implies} y=0$, i.e., the volunteered posts are non-damaging by definition. For ease of notation, let $\Nofferedlt := |\cD_{\leq t}^\offeredset|$ be the number of volunteered posts collected till interval $t$.

Then, the goal of the challenger is to construct the decoy set $\sG^\decoyset_{t+1} \subseteq \cD_{\leq t}^\offeredset$ and delete these posts during the next time interval $t{+}1$ in order to fool the adversary into misclassifying these challenger-deleted non-damaging posts as user-deleted damaging posts.
Formally, we want to choose $K$ \decoy posts (denoted by a $K$-hot vector $\vw$) that maximizes the negative-log likelihood loss for the adversary's classifier, given by the following optimization problem,
\begin{align} \label{eq:optimization_problem}
\vw^* &= \argmax_{\vw} V(\vw; \cD_{\leq t}^\offeredset) \nonumber \\
\text{s.t. ~~~} & ||\vw||_1 = K, ~~~~~~ \vw \in \{0, 1\}^{\Nofferedlt} \:, 
\end{align}
\vspace{-2mm}
\noindent where
\vspace{-1mm}
\begin{align} \label{eq:V}
V(\vw; \cD_{\leq t}^\offeredset) &= \sum_{i=1}^{\Nofferedlt}  {-}\evw_i \cdot \log(1 - \adv(x_i; \vtheta_t)) \:,
\end{align}
and $x_i$ is the $i$-th volunteered post in $\cD_{\leq \dt}^\offeredset$.
The cost function $V(\vw; \cD_{\leq t}^\offeredset)$ in \cref{eq:V} is simply the negative log-likelihood of the adversary over the set $\cD_{\leq t}^\offeredset$ weighted by a $K$-hot vector $\vw$. \cref{eq:V} uses the fact that the set only contains  non-damaging posts (i.e., $y_i = 0$).

Consequently, $\vw^*$ optimized in such a fashion selects $K$ posts from the set $\cD_{\leq t}^\offeredset$ that \textit{maximizes} the adversary's negative log-likelihood loss. 
The set of $K$ selected posts can be trivially constructed as $\sG_{t+1}^* = \{x_i : i \in \{1, \ldots, \Nofferedlt \}\wedge w_i = 1\}$. 
The challenger deletes $\sG_{t+1}^*$ over the next time interval $t{+}1$ (hence the adversary sees these posts as part of the deleted set $\cD_{\dt+1}^\delta$). 
Note that the challenger uses the adversary's classifier $a(~\cdot~;~\theta_t)$ to create decoy posts for $t{+}1$. However, as per \cref{subsec:adversary}, in interval $t{+}1$ the adversary first trains over a sample of the deleted posts  (including the decoy posts) and updates its classifier to $a(~\cdot~; \theta_{t+1})$ before classifying the rest of the deleted posts of $t{+}1$. Hence, the challenger is always at a disadvantage (one step behind).

Next, we describe three \confusers corresponding to the access types discussed in \cref{subsec:threat_challenger}: \textit{no access}, \textit{black-box access} and \textit{monitored black-box access with a query budget}.

\begin{algorithm}[t]
	\SetKwInOut{Input}{input}\SetKwInOut{Output}{output}
	\SetAlgoLined
	
	\Input{ $\setDecoyPostsNotime, \; K, \; $ accessType}
	
	$\sG^\decoyset \leftarrow \emptyset$ \;
	
	\uIf{accessType = none}{
	\tcc{Random challenger}
			$\sG^\decoyset~\sample{K}~\cD^\offeredset$ \;
	}
	\vspace{2mm}
	\uElseIf{accessType = black-box}{
	    \tcc{Oracle challenger}
		$\sG^* \leftarrow \{x_i : x_i \in \setDecoyPostsNotime \wedge \adv(x_i; \vtheta) \text{ is in the top } K \}$ \;
	}
	\vspace{2mm}
	\uElseIf{accessType = monitored black-box (budget $\budgetC$)}
	{
	    \tcc{\system challenger}
        Sample $\budgetC$ posts for training $\cD^{\offeredset,\subtrain} ~\sample{\budgetC}~\cD^\offeredset$\;
        $\cD^{\offeredset,\subtest} \leftarrow \cD^\offeredset \setminus \cD^{\offeredset, \subtrain}$  \;
		Query $\adv(x_i; \vtheta)$ for all $(x_i, y_i = 0) \in \setDecoyPostsTrainNotime$ \; 
		Obtain optimal parameters $\vphi^*$  by solving~\cref{eq:continuous_relaxation} \;
		$\sG^* \leftarrow \{x_i : x_i \in \setDecoyPostsTestNotime \wedge \con(x_i; \vphi^*) \text{ is in the top } K \}$ \;
	}
	\vspace{2mm}
	\Return $\sG^\decoyset$  \;
	
	\caption{Challenger \label{alg:challenger}}
\end{algorithm}

\paragraph{Random challenger (no access)}
We begin with the case where the challenger has \textit{no access} to the adversary's classifier and there is no side-information available to the challenger.
With no access to the adversary's classification probabilities $a(~\cdot~; \vtheta_t)$, the optimization problem in \cref{eq:optimization_problem} cannot be solved.
We introduce the naive \textit{random challenger} that simply samples $K$ posts randomly from the available volunteered posts $\cD_{\leq t}^\offeredset$ and deletes them, i.e., $\sG_{t+1}^*~\sample{K}~\cD_{\leq t}^\offeredset$. 
This is the only viable approach if the \confuser has no information about the adversary's classifier.

\paragraph{Oracle challenger (black-box access)}
Next we consider the challenger that has a black-box access to the adversary's classifier with no query budget, i.e., at any time interval $t$, the challenger can query the adversary with a post $x$ and expect the adversary's predicted probability $a(x; \vtheta_{t})$ in response without the adversary's knowledge. 
Armed with the black-box access, oracle \confuser can simply maximize \cref{eq:optimization_problem} by choosing the top $K$ posts with highest values for $\adv(x_i; \vtheta_{t})$.

\vspace{1mm}
\paragraph{\system \Confuser (monitored black-box access with query budget $\budgetC$)}
The oracle challenger assumes an \textit{unmonitored} black-box access to the adversary \textit{with an infinite query budget} which can be hard to obtain in practice. In what follows, we relax the access and assume a \textit{monitored} black-box access \textit{with a recurring query budget of $\budgetC$}. In other words, queries to the adversary, while being limited per interval, are also monitored and possibly flagged by the adversary. The adversary can simply take note of these queries as performed by a potential challenger, hence negating any privacy benefits from injecting \decoy posts. Whenever the adversary sees a \textit{deleted post} identical to one that it was previously \emph{queried} about, it can ignore the post as it is likely non-damaging. 

Here we design a challenger, henceforth dubbed \system, that \textit{trains to select decoy posts} from any given volunteered set. In other words, the \system challenger makes use of the monitored black-box access to the adversary only during training. Hence it can be used to find the decoy posts {without} querying the adversary; for example in a held-out volunteered set (separate from the training set). Additionally, the \system challenger queries the adversary for only $\budgetC$ posts every \epoch.

We denote the challenger's model at the beginning of interval $\dt$ by $\con(~\cdot~; \vphi_{t{-}1}) : \sX \rightarrow \sR$ parameterized by $\vphi_{t{-}1}$. For a given volunteer post $x$, $g(x; \vphi_{t{-}1})$ gives an unnormalized score for how likely the post will be mislabeled 
as damaging; higher the score, higher the misclassification probability. 

First, the \system challenger samples $\budgetC$ posts for training from the available volunteered set $\cD_{\leq t}^{\offeredset}$ collected till interval $t$. We denote the train and test sets of the \system challenger as $\cD_{\leq t}^{\offeredset, \subtrain}$ and  $\cD_{\leq t}^{\offeredset, \subtest}$ of sizes $\budgetC$ and $\Nofferedlt - \budgetC$ respectively. 
Then, the goal of the \system 
is to find optimal parameters $\vphi_{t}$ by solving a continuous relaxation of \cref{eq:optimization_problem} presented below,
\begin{align}\label{eq:continuous_relaxation}
\vphi_{t} = \argmax_{\vphi} \tilde{V}(\vphi; \setDecoyPostsTrain)
\end{align}

\vspace{-2mm}
\noindent where
\begin{align}
\tilde{V}(\vphi; \setDecoyPostsTrain) &= \sum_{i=1}^{\NofferedltTrain} {-} \alpha(x_i; \vphi, \setDecoyPostsTrain) ~ \log(1 - \adv(x_i; \vtheta_t))\:, \nonumber
\end{align}
\noindent and
\begin{align*} \label{eq:alpha}
\alpha(x_i; \vphi, \setDecoyPostsTrain) &= \frac{\exp{(\con(x_i; \vphi))}}
{\sum_{j=1}^{\NofferedltTrain} \exp{(\con(x_j; \vphi))}} \:,
\end{align*}
is a softmax over the challenger outputs for all the examples in $\setDecoyPostsTrain$. The softmax function makes sure that $0 \leq \alpha(~\cdot~; \vphi,  \setDecoyPostsTrain) \leq 1$ and $\sum_{j=1}^{\NofferedltTrain} \alpha(x_j; \vphi,  \setDecoyPostsTrain)  = 1$. The continuous relaxation in \cref{eq:continuous_relaxation} allows the \system challenger to train a neural network model parameterized by $\phi$ via backpropagation.

We now show that optimizing the relaxed objective in \cref{eq:continuous_relaxation} results in the best objective value for \cref{eq:optimization_problem}.
\begin{proposition}\label[proposition]{prop:confuser}
For any given volunteered set $\cD^\offeredset$ with $N$ non-deleted posts,
\vspace{-2mm}
$$
\max_{\vphi} \tilde{V}(\vphi; \cD^\offeredset) = \max_{w_1, \ldots, w_{N}} V(w_1, \ldots, w_{N}; \cD^\offeredset)
$$
\end{proposition}
\vspace{-2mm}

We present proof of the proposition in~\cref{app:proofs}. 

Finally, the \system challenger with optimal parameters $\vphi_{t}$ computes $g(x; \vphi_{t})$ for all $(x, y=0) \in \cD_{\leq t}^{\offeredset, \subtest}$, and constructs $G_{t+1}^{*}$ by choosing the examples with top $K$ values for $g(~\cdot~; \vphi_{t})$. 
\cref{alg:challenger} shows the actions of the \confuser within a time interval (subscript $t$ removed for clarity).

\subsection{Deceptive Learning Game}
\cref{alg:game} presents the game between the adversary and the challenger. 
In each time interval, users independently delete and volunteer posts (line 4).
The platform/deletion-service additionally deletes the challenger-selected decoy posts (line 5).
The adversary obtains all the deleted posts and queries the MTurk with a small subset of the posts for labels (if the adversary has not exhausted the budget). With this labeled set of deleted posts, the adversary trains its classifier (lines 6-7). 
The challenger collects new volunteered posts (line 8) and builds decoy posts to be injected in the next interval (line 9). This results in a real-life game between the adversary and the challenger, where each adapts to the other. The \game is different from the adversarial learning approaches as we detail in \cref{sec:adversarial_learning}.

\begin{algorithm}[t]
	\SetKwInOut{Input}{input}\SetKwInOut{Output}{output}
	\SetAlgoLined
	\Input{~ accessType, $K$}
	$\sG_1^\decoyset \leftarrow \emptyset $ \;
	$\cD_{\leq 0}^\offeredset \leftarrow \emptyset$ \;

	\For{$t \leftarrow$ 1 \KwTo n }{
		
 		$\cD_\dt^\delset, ~\cD_\dt^\offeredset \leftarrow $ Users($t$) \tcc*[r]{deleted and volunteered posts of the users at interval $t$} 
	
		\vspace{1mm}
		
		$\cD_\dt^\delset \leftarrow \cD_\dt^\delset \cup \sG_\dt^\decoyset$ \tcc*[r]{user- and \confuser-deleted posts at interval \dt}

		\vspace{1mm}
		\uIf{Adversary's budget has not exhausted}
		{$\adv(~\cdot~,\vtheta_t)  \leftarrow \text{Adversary}(\cD_\dt^\delset)$ \;}
	    
	    \vspace{1mm}
	    $\setDecoyPosts \leftarrow (\cD^\offeredset_{\leq t-1} \setminus \sG_\dt^\decoyset) \cup \cD_\dt^\offeredset $
	    \tcc*[r]{available volunteered set}
		
		\vspace{1mm}
		$\sG_{\dt+1}^\decoyset \leftarrow \text{\Confuser} (\setDecoyPosts, K,  \text{accessType})$
	}
	\caption{Deceptive Game \label{alg:game}}
\end{algorithm}

\subsection{Analysis: Who Wins the Game?}\label{sec:whowins}
In what follows, we analyze the scenarios where either the adversary or the challenger wins the \game. 
We show that the volunteered set, $\cD^\offeredset$, plays a significant role in deciding the winner of the game. First, we need the definition of support of a distribution.
\begin{definition}[Support]\label{def:support}
	Let $\Omega = \{x : \forall x, p(x) > 0\}$ be the support of  distribution $p(x)$, i.e., the set of all possible features $x$ with non-zero probability.
\end{definition}

Let $p^\damset(x)$ be the distribution of the features of damaging posts, with the corresponding support denoted by $\Omega^\damset$. Then, a post $x$ is in $\Omega^\damset$ if there is a non-zero probability that it is a damaging post. Similarly, $\Omega^\offeredset$ is the support of the distribution of volunteered posts $p^\offeredset$.
Next, we analyze the two extreme scenarios of non-overlapping supports (i.e., $\Omega^\offeredset \cap \Omega^\damset = \emptyset$) and fully-overlapping supports (i.e., $\Omega^\offeredset = \Omega^\damset$). These extreme scenarios correspond to the following simple questions respectively: (a) ``\textit{what if all the posts volunteered by users have completely different features than the damaging posts?}'' and (b) ``\textit{what if the volunteered posts have very similar or same features as those of damaging posts?}''.

\begin{subsubsection}{Non-overlapping Support: Adversary Wins}

\begin{proposition}[{\em Non-overlapping} support]\label[proposition]{p:nonoverlapping}
	Assume $\Omega^\offeredset \cap \Omega^\damset = \emptyset$, i.e., the supports of volunteered and damaging posts do not overlap.
	Then, there is always a powerful-enough \adversary to defeat the \confuser. 
\end{proposition}
\begin{proof}[Proof sketch]
	Assume the most powerful \confuser who can select any post features $x$ from an infinite supply of volunteered posts. However, since $\Omega^\offeredset \cap \Omega^\damset = \emptyset$, there is no sampling from $p^\offeredset$ to generate decoy examples that look like they are sampled from $p^\damset$.  Hence, given enough data, an \adversary can find a perfect decision boundary between the damaging posts and the decoy posts. Because neural networks are universal function approximators~\cite{hornik1989multilayer}, this powerful \adversary always exists and, thus, the \confuser can always be defeated in the \game.
\end{proof}

\textit{An Illustrative Example: } Consider the example provided in \cref{fig:halfmoon_adv_decision}. The two classes (denoted by red circles and green crosses respectively) have non-overlapping support. We show the decision boundary of the adaptive adversary in this setting dataset after 50 intervals of the \game. We see that the adversary can perfectly label the points even in the presence of the oracle challenger. 

\textit{Real-world scenario: } The non-overlapping case could happen in an online social platform if its users are {very conservative} in volunteering posts to the challenger. Consider for example, none of the volunteered posts contained {any} sensitive keyword, whereas all the damaging posts had at least one sensitive keyword, a clear case of non-overlapping supports. In such a scenario, the adversary will win the game as detailed above. 

\end{subsubsection}

\begin{subsubsection}{Fully-overlapping Support: Challenger Wins}

\begin{proposition}[Fully overlapping support]\label[proposition]{p:fullyoverlapping}
	Assume $\Omega^\offeredset=\Omega^\damset$, i.e., the supports of volunteered and damaging posts fully overlap.
	Then, given enough volunteered posts in $\cD^\offeredset$, the \confuser  always defeats the \adversary (in both static and adaptive scenarios).
	More precisely, if the \confuser selects $k$ decoys per damaging post in $\cD^\delset$, then the  \adversary's probability of identifying a damaging post in $\cD^\delset$ is in average at most $1/(k+1)$. 
\end{proposition}
\begin{proof}[Proof sketch]
	The proof relies on a property of rejection sampling, which states that if the support of two distributions $p_1$ and $p_2$ fully overlap, then one can selectively filter samples from $p_1$ to make the filtered samples have distribution $p_2$ (a proof of this principle is given in the Appendix). 
	Asymptotically, for each damaging example $x$ in adversary's test data, there are $k$ indistinguishable decoy examples (from the adversary's perspective).
	This is because, by Bayes theorem 
	\begin{align*}
	 &p^\delset(y=1|x) =\\
	&\frac{p^\delset(x | y =1) p^\delset(y=1)}
	{p^\delset(x | y =1)p^\delset(y=1) + p^\delset(x | y =0)p^\delset(y=0)}
	\leq \frac{1}{1 + k},
	\end{align*}
	where the superscript $p^\delset$ indicates the distribution of deleted posts $\cD^\delset$.
	The inequality holds by construction, as for all $x \in \cD^\delset$ with label one, there are at least $k \geq 1$ samples from $p^\offeredset(x)$ with label zero.
\end{proof}

\textit{An Illustrative Example: } Consider the example provided in \cref{fig:gaussian_adv_decision} where the two classes (red circles and green crosses respectively) have fully overlapping supports (as they are drawn from a Gaussian distribution with different means). We show the decision boundary of the adaptive adversary in this setting after 50 intervals of the \game. We see that for any decision boundary, there exist points in $\Omega^\offeredset$ that a challenger can choose such that the adversary mislabels them as damaging. 

\textit{Real-world scenario: } The fully-overlapping case could happen in an online social platform if the definition of what constitutes as damaging varies across the platform's users. For example, user $A$ could consider a post with a single sensitive word (e.g., a swear word) as damaging, whereas another user $B$ from a different background might consider the same post as completely innocuous and volunteer the post. In such a scenario, the challenger will use volunteered posts from user $B$ to protect the damaging posts of user $A$. Hence, the challenger will win the game against even the most powerful adversary with infinite data.

\end{subsubsection}

\cref{p:fullyoverlapping,p:nonoverlapping} are important to understand the two extreme cases ---where either the \confuser clearly wins or the \adversary clearly wins--- as important insights, even though these clear-cut cases are unlikely to happen in practice. 
Most real-world applications will likely fall between these two extremes, where the supports only partially overlap. 
In such scenarios, the adversary wins outside the overlap (i.e., can classify everything correctly outside the overlap), and the challenger wins inside the overlap. In other words, extremely sensitive and damaging posts cannot be protected as they will have no overlap with any of the volunteered posts. However, as we show in the next section, with a reasonable volunteered set, the challenger can make it hard for the adversary to detect damaging deletions.

\begin{figure}[t]%
	\begin{subfigure}{0.44\columnwidth}
		\includegraphics[height=1.4in]{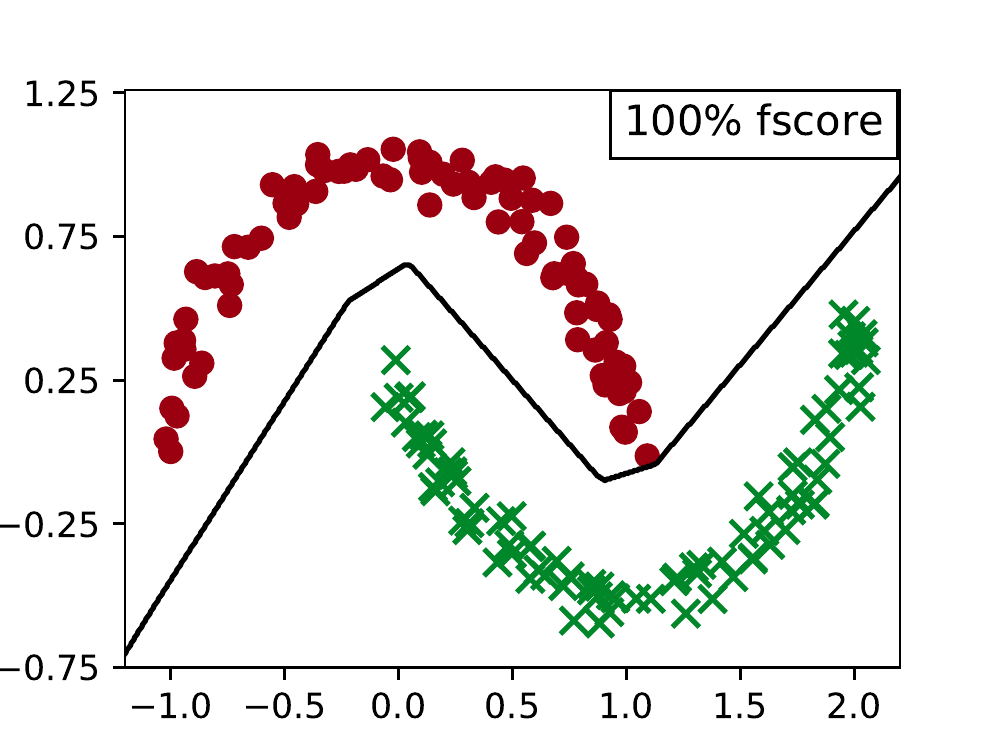}
		\caption{Non-overlapping supports\label{fig:halfmoon_adv_decision} }
	\end{subfigure}
	~~~~~
	\begin{subfigure}{0.44\columnwidth}
	\includegraphics[height=1.4in]{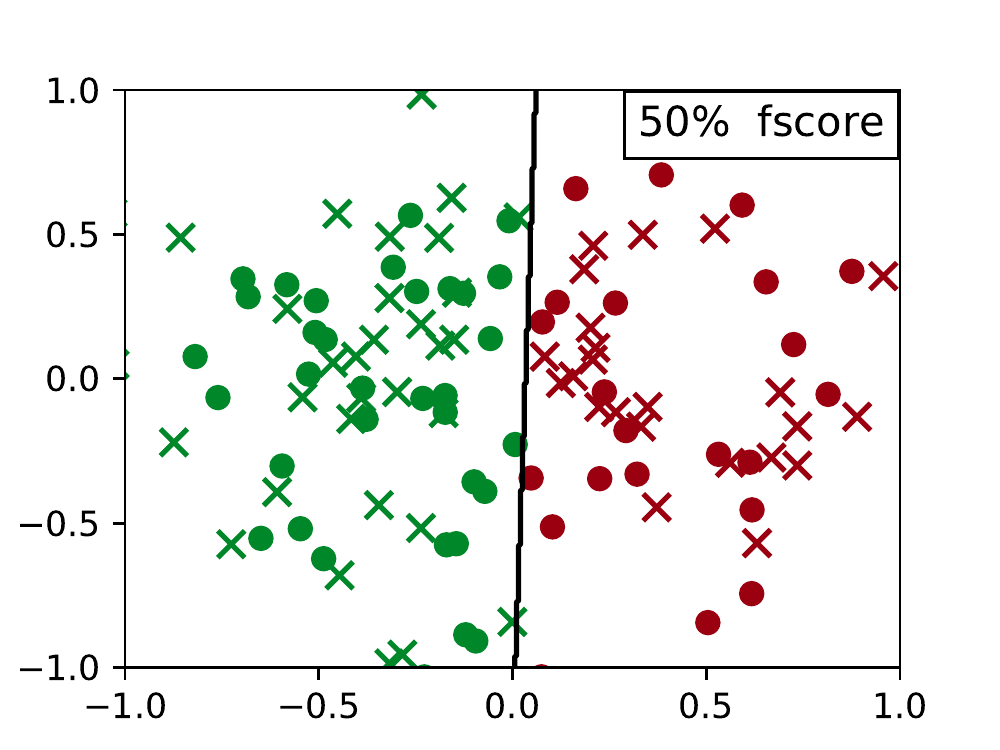}
	\caption{Fully overlapping supports\label{fig:gaussian_adv_decision} }
	\end{subfigure}
    \caption{Two examples illustrating the two possible scenarios relating to the supports of volunteered posts and damaging posts: non-overlapping (left) and fully overlapping (right). The black line denotes the decision boundary of the adaptive adversary after 50 intervals of the \game. \label{fig:adv_decision} }
    \vspace{-3mm}
\end{figure}

%% file: twitter_experiment.tex
\section{System Evaluation on Twitter Deletions}\label{sec:exp}

In this section we evaluate the efficiency of an adversary when \sysname is applied to the real-world problem of concealing damaging deletions in Twitter.
In this evaluation we first create and prepare sets of (non)damaging tweets. Then we use these sets to train the \confuser and adversary classifiers and analyze their performance.

\subsection{Data Collection} \label{subsec:data}
In this work, we select Twitter as our experimental social media platform.
We note that it was certainly plausible to perform the exact experiment on other social platforms. However we chose Twitter due to its popularity and feasibility of data collection.
Specifically, in order to evaluate the \confuser we needed a real-world dataset which includes (i) both deleted and non-deleted tweets (i.e., Twitter posts) and (ii) deleted tweets that contain both damaging and non-damaging tweets. 
To that end, we use two data sources to create such a dataset. 

\subsubsection{Deceptive Deletion dataset}
We collected 1\% of daily random tweet samples from the Twitter API from Oct 2015 - May 2018. 
Eliminating  non-English tweets, we accumulated over one billion tweets. 
In the next step, we construct the damaging and volunteered sets.

To construct the damaging set, we first needed to identify the deleted tweets.
We sampled 300,000 tweets from the aforementioned collected data, and leveraging the Twitter API, we identified the tweets that were deleted at the time of our experiment (Jan 30th, 2020). 
In total, we identified 92,326 deleted tweets.
The next step was to obtain ground truth labels for the deleted tweets---i.e., detect and assign ``true'' labels to damaging tweets and ``false'' labels to rest.
We used the crowdsourcing service Amazon Mechanical Turk (MTurk)~\cite{amt}  to obtain a proxy for these true labels. 
However, there were two challenges-- First, it was impractical to ask our annotators to label 92,326 tweets. Second, since the dataset was highly imbalanced, a simple random sample of tweets for labeling would have resulted in a majority of non-damaging tweets. 

Thus, we followed prior work to create a more balanced sample dataset~\cite{wang2019donttweetthis, zhou-2016-regrettableDelTweet}. 
Specifically, we filtered the deleted tweets using a simple sensitive keyword-based approach~\cite{zhou-2016-regrettableDelTweet} (i.e., identify posts with sensitive keywords) to have a higher chance of collecting possibly damaging tweets.
This approach resulted in 33,000 potentially damaging tweets, and we randomly sampled 3,500 tweets
to be labeled by annotators on MTurk.
The mean number of sensitive keywords in each tweet within our data set was $2.55$. 

Note that, in addition to the cursing and sexual keywords, our sensitive keyword-based approach also considered keywords related to the topics of religion, race, job, relationship, health, violence, etc. Intuitively, if a post does not contain any such sensitive keywords then the likelihood of the post being damaging is very low. We confirmed this intuition by asking MTurk annotators to label 150 tweets which did not contain any sensitive keyword as damaging/non-damaging. We noted that more than 97\% of these 150 tweets were labeled as non-damaging by annotators.
We surmised that in practice, the adversary will also leverage a similar filtering approach to reduce its overhead and increase its chances of finding damaging posts.

In total, out of our sampled 3,500 deleted tweets, we obtained labels for 3,177 tweets (excluding annotations from Turkers who failed our quality control checks as described later). 
Among the labeled tweets, 1,272 were identified as damaging, and $1,905$ were identified as non-damaging.

\paragraph{Data labeling using MTurk}
We acknowledge that ideally, the tweet labels should have been assigned by the posters themselves. 
However, since we collected random tweets at large-scale using the Twitter API, 
we could not track down and pursue original posters to label their deleted tweets. 
To that end, we note that there is a crowdsourcing based alternative which is already leveraged by earlier work to assign sensitivity labels~\cite{correa-2015-anonymityShades,biega-2016-rsusceptibility,wang2019donttweetthis}. 
Specifically, these studies determined the sensitivity of social media posts by simply aggregating crowdsourced sensitivity labels provided by multiple MTurk workers (Turkers). Thus we took a similar approach as mentioned next.

On MTurk, tasks (e.g., completing surveys) are called Human Intelligence Tasks or HITs. 
Turkers can participate in a survey by accepting the corresponding HIT only if they meet all the criteria associated with that HIT (set by the person(s) who created the HIT). We leverage this feature to ensure the reliability of our results. Specifically we asked that the Turkers taking our survey should: 
(i) have at least 50 approved HITs.
(ii) have an assignment approval rate higher than 90\%, and 
(iii) have their location set to United States. 
This last criterion ensured consistency of our Turkers' linguistic background. 
In our experiment each HIT consisted of annotating 20 tweets with true (damaging) or false (non-damaging) labels. 
We allowed the Turkers to skip some tweets in case they feel uncomfortable for any reason. 
We compensated 0.5 USD for each HIT and on average it took the Turkers 193 seconds to complete each HIT. 

To control the quality of annotation by Turkers, we included two hand-crafted control tweets with known labels in each HIT. These control tweets were randomly selected from two very small sets of clearly non-damaging or damaging tweets and were inserted at random locations within the selection of 20 tweets.
For example a damaging control tweet was: \textit{``I think I have enough knowledge to make a suicide bomb now! Might need it New Year's Eve"} and non-damaging control tweet was: \textit{``Prayers with all the people in the hurricane irma"}. 
If for a HIT, the responses to these control tweets did not match the expected label, we conservatively discarded all twenty annotations in that HIT.

We countered possible bias resulting from the order of presentation of tweets via randomizing the order of tweets in every HIT. 
Even if two Turkers annotated the same set of tweets, the order of those tweets was different. 
Furthermore, to ease the subjectivity of the labels from each participant, for each tweet we collected the annotations of multiple Turkers 
and took the majority vote.
In our experiment, we created the HITs such that each tweet was annotated by 3 distinct Turkers.
After receiving the responses, for each tweet we assigned the final label (indicating damaging or non-damaging) based on the majority vote.

We emphasize that in the real world, the burden of labeling the posts via crowdsourcing is on the adversary. The challenger, on the other hand, can be implemented as a service within the platform and can obtain the true labels directly from the post-owners. Therefore, existence of any mislabeled data will negatively impact only the adversary.

\subsubsection{\#Donttweet dataset}
Recently Wang et al.~\cite{wang2019donttweetthis} proposed ``\#Donttweetthis''. 
``\#Donttweetthis'' is a quantitative model that identifies potentially sensitive content and notifies users so that they can rethink before posting those content on social platforms.
Wang et al. created the training data for their model by (i) identifying possibly sensitive tweets by checking for the existence of sensitive keywords within the text and then (ii) using crowd-sourcing (i.e., using MTurk) to annotate the sensitivity of each tweet by three annotators. 

The data collection approach used by ``\#Donttweetthis'' (section 3 of~\cite{wang2019donttweetthis}) is very similar to ours. 
Therefore, to enrich our dataset and be able to evaluate the \confuser over more intervals, we acquired their labeled tweets. Using the Twitter API, we queried the tweets using their corresponding IDs and identified the deleted ones (at the time of writing, Jan 30th, 2020).
In total, we obtained $851$ deleted tweets, where $418$ were labeled as sensitive (damaging), and the remaining $433$ were labeled as non-sensitive (non-damaging). 
The mean of sensitive keywords in each tweet within this set was $1.7$. 

\paragraph{Summary of collected data}
In summary, combining the two datasets explained above, we obtained labels for $4,028$ deleted tweets establishing the user deleted set. Among the deleted tweets $1,690$ were labeled as damaging constructing our damaging set ($\cD^\damset$). 
As we will demonstrate  in the results section, in our evaluation the four thousand labeled tweets (larger than that of prior works~\cite{wang2019donttweetthis,zhou-2016-regrettableDelTweet}) allows for $10$ intervals for the game between the adversary and challenger.

Furthermore, for our experiment, we consider $k = 1,2,5$ (i.e., number of decoy posts for each damaging post).
We needed to accommodate these values of $k$ and also a  pool that the challenger can make meaningful selections from.  
 Thus, we sampled $100,000$ non-deleted tweets uniformly at random from the 1\% of daily random tweet sample posted between Jan 1st, 2018 -- May 31st, 2018 to build the volunteered set.
 The non-deleted tweets are assumed to be non-damaging. We consider this assumption to be reasonable as if a tweet contains some damaging content then its owner would not keep that post on its profile.
 In practice, we can forgo this assumption as the volunteer users themselves offer the volunteer posts.
The average number of sensitive keywords in each tweet in this set was $0.41$.

\subsection{Ethical Considerations}
	Recall that in order to create our evaluation dataset we needed to show some deleted tweets to Turkers for the annotation task. Thus, we were significantly concerned about the ethics of our annotation task. Consequently, we discussed at length with the Institutional Review Board (IRB) of the lead author's  institute and deployed the annotation task only after we obtained the necessary IRB approval. Next we will detail, how, in our final annotation task protocol we took quite involved precautionary steps for protecting the privacy of the users who deleted their tweets.

	We recognize that, in the context of our evaluation, the primary risk to the deleted-tweet-owners was the possibility of linking deleted tweets with deleted-tweet-owner profiles during annotation. This intuition is supported by prior research~\cite{maddock2015using,mondal-2016-longitudinal-exposure} who suggested applying selective anonymization for research on deleted content. Thus, we anonymized all deleted tweets by replacing personally identifiable information or PII (e.g., usernames, mentions, user ids, and links) with placeholder text.
	For example, we replaced user accounts (i.e., words starting with @) and url-links with ``UserAccount" and ``Link'' respectively.
	Moreover, one of the authors manually went over each of these redacted posts to ensure anonymization of PII before showing them to Turkers.

\subsection{Experiment Setup}\label{sec:exp_setup}

\paragraph{Partitioning the data for different time intervals}
Recall from \cref{sec:sys_threat_model} that we discretize time into intervals. In our experiments, we choose $T=10$ intervals in total (a choice made based on the number of collected tweets). Consequently, we partition our dataset into 10 intervals. Ideally, the partitions should be based on the creation and deletion timestamps of the tweets. 
Unfortunately however, the Twitter API does not provide deletion timestamps. Hence, we randomly shuffle the tweets and divide them into 10 equally sized partitions. 

\vspace{1mm}
\paragraph{BERT model} 
In line with our approach to model the most-powerful adversary as best as we possibly can, we use a state-of-the-art natural language processing model: the BERT (Bidirectional Encoder Representations from Transformers) language model~\cite{devlin2018bert}, both for the adversary and for the challenger. Specifically, we use $\text{BERT}_\text{BASE}$ model that consists of 12 transformer blocks, a hidden layer size of 768 and 12 self-attention heads (110M parameters in total). BERT has been shown to perform exceedingly well in a number of downstream NLP tasks \cite{devlin2018bert}.
We use HuggingFace's \cite{Wolf2019HuggingFacesTS} implementation of the BERT model that was already \textit{pre-trained} on masked language modeling and next sentence prediction tasks.

BERT uses WordPiece embeddings \cite{wu2016google} to convert each word in the input tweet to an embedding vector. The concatenated embedding vector is passed to the BERT neural network model. In our experiments, we only give the text of the tweet as input to both the adversary and the challenger to make it amenable to the pre-trained BERT models. Other tweet features such as deletion timestamps, number of likes, etc. could be used by both the adversary and the challenger to improve their performance. Note however that \cref{p:nonoverlapping,p:fullyoverlapping} still apply as long as the adversary and the challenger have the same information.

We \textit{fine-tune} the BERT model on our datasets as prescribed by Devlin et al \cite{devlin2018bert}. In each interval, the adversary's classifier is fine-tuned for the classification of tweets into damaging and non-damaging using the negative log-likelihood loss in \cref{eq:adversary}. We use a batch size of 32 and sample equal number of damaging and non-damaging tweets in each batch. This procedure results in better trained models as it avoids the scenario where a randomly sampled batch is too imbalanced (for example, no damaging tweet sampled in the batch).
A separate BERT model is fine-tuned for the challenger using the loss function in \cref{eq:continuous_relaxation}. Note that no balancing is required here since all the input tweets to the challenger model are non-damaging.
We note that explaining the exact strategy employed by BERT models to classify text is an active research topic and complementary to our efforts. However, we highlight that our challenger does not use any information about either the adversary's exact model or its parameters.

\emph{Budget constraints: }
We allow a limited budget of $\budgetSA=200$ deleted tweets for the static adversary and set $\tau=1$, i.e., the static adversary only trains during the first out of the ten intervals. Similarly for the adaptive adversary, we allow a fixed recurring budget of $\budgetAA=200$ deleted tweets every interval. 
There are no budget restraints for random and oracle challengers (having no access and black-box access respectively). However, we restrict the \system challenger to have the same (recurring) query budget as the adaptive adversary's recurring budget to keep the game fair, i.e., $\budgetC = \budgetAA = 200$. 

We simulate the \game described in \cref{alg:game} with an adversary and a challenger, both implemented as BERT language models, with 10 different random seeds. 
We repeat the experiments for $k=1, 2, 5$ where $k$ denotes the number of decoy posts added per damaging deletion.

%% file: results.tex
\subsection{Results} \label{sec:results}
\cref{fig:twitter_results_adversary,fig:entire_results_figure} show the F-scores (with 95\% confidence intervals), precision and recall for different adversaries over 10 time intervals.
\noindent We make the following key observations from the results.

 \begin{figure}
	\centering
	\includegraphics[scale=0.7]{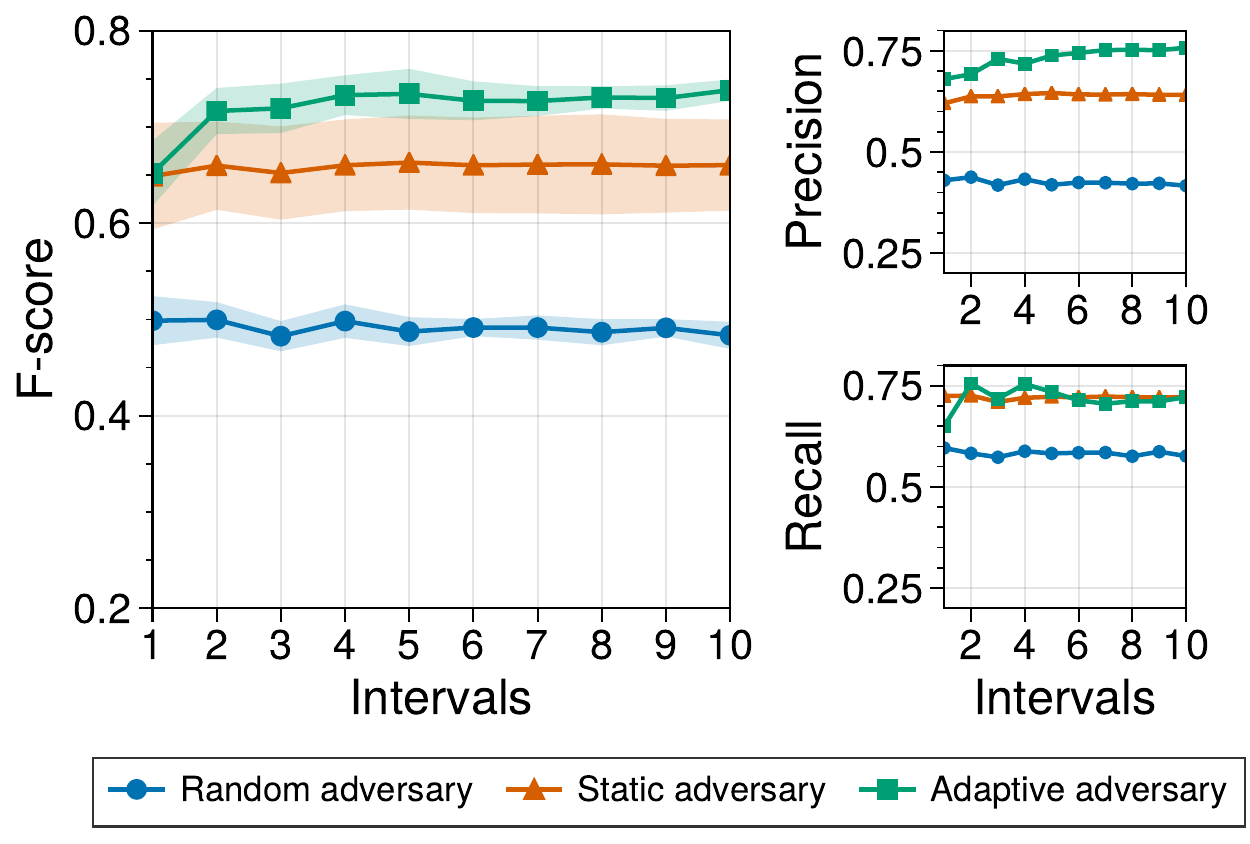}
	\caption{F-score of different adversaries (random, static, adaptive) when no privacy preserving deletion mechanism is in place. Shaded areas represent 95\% confidence intervals.}
	\label{fig:twitter_results_adversary}
\end{figure}

\newcommand{\figurescale}{0.45}
\newcommand{\figurewidth}{0.32}

\begin{figure*}[!ht]
    \centering
    
    \begin{subfigure}{0.5\textwidth}
        \hspace{-10mm}
		\includegraphics[scale=0.4]{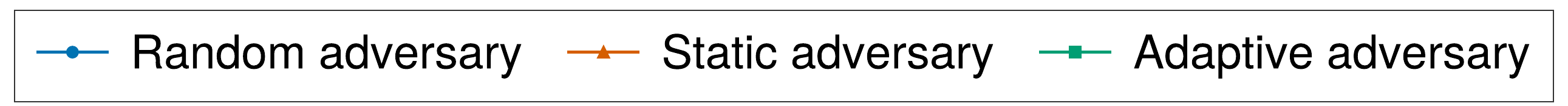}
	\end{subfigure}
	\vspace{0.2in}
 \begin{subfigure}{\textwidth}
 	\centering
     \begin{subfigure}{\figurewidth\textwidth}
 		\includegraphics[scale=\figurescale]{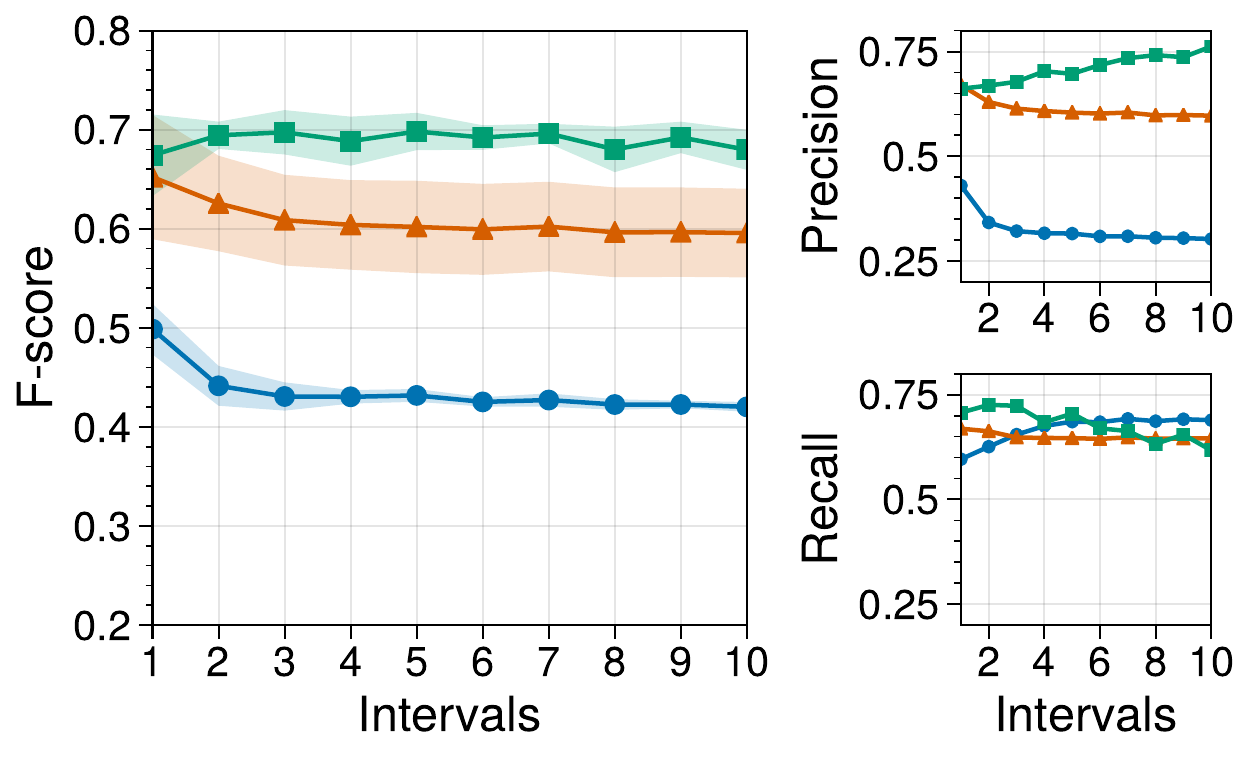}
 		\caption{Random challenger ($k = 1$)}
 	\end{subfigure}
 	\begin{subfigure}{\figurewidth\textwidth}
 		\includegraphics[scale=\figurescale]{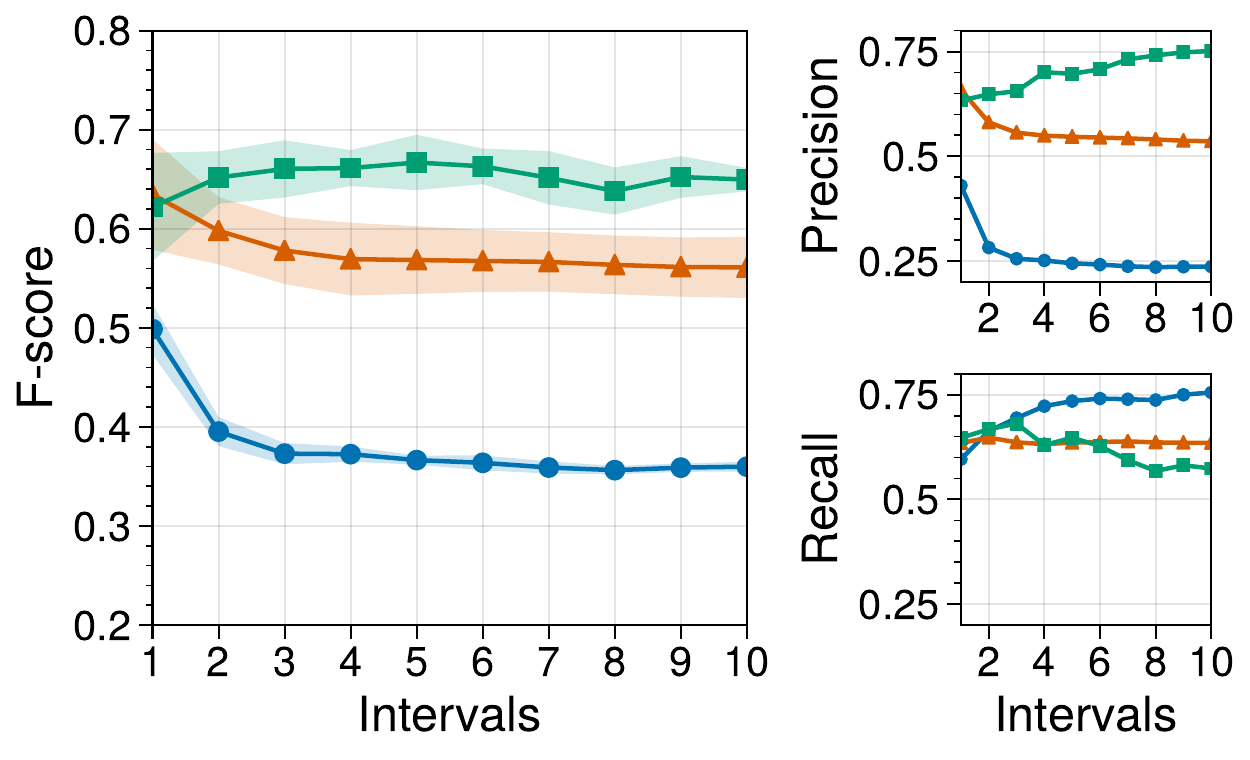}
 		\caption{Random challenger ($k = 2$)}
 	\end{subfigure}
 	\begin{subfigure}{\figurewidth\textwidth}
 	\includegraphics[scale=\figurescale]{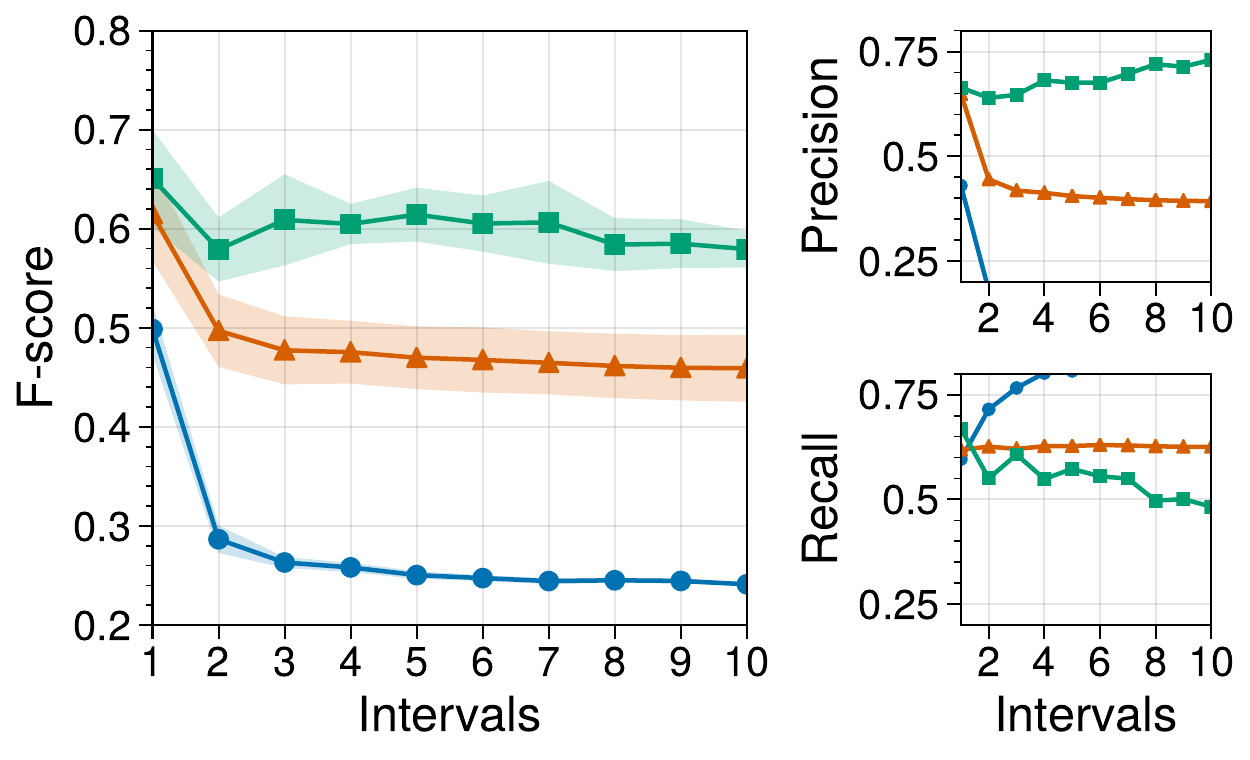}
 	\caption{Random challenger ($k = 5$)}
 	\end{subfigure}
 	\caption*{\textbf{(No access.)} Adversaries (random, static and adaptive) in the presence of \textbf{random challenger} with $k=1,2,5$.}
 	\label{fig:results_more2_random}
 \end{subfigure}

\vspace{0.1in}
 \begin{subfigure}{\textwidth}
 	\centering
 	\begin{subfigure}{\figurewidth\textwidth}
 		\includegraphics[scale=\figurescale]{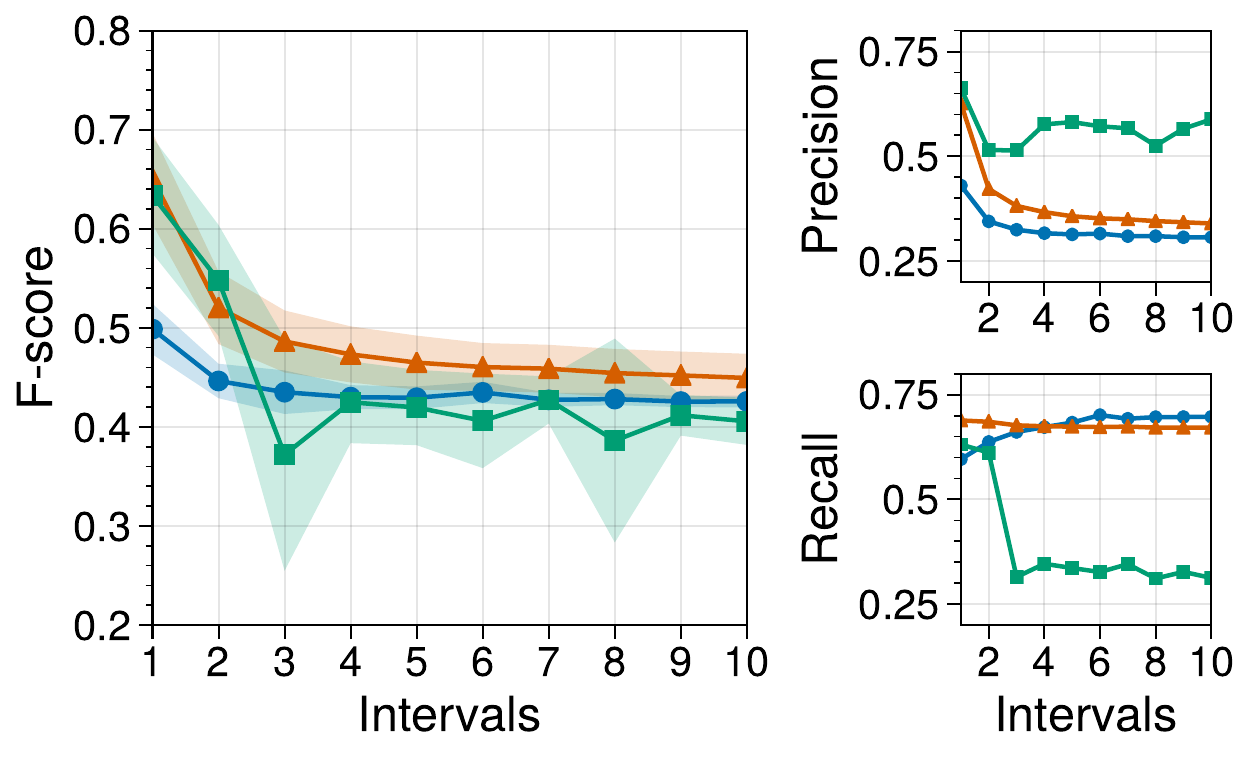}
 		\caption{Oracle challenger ($k = 1$)}
 	\end{subfigure}
 	\begin{subfigure}{\figurewidth\textwidth}
 		\includegraphics[scale=\figurescale]{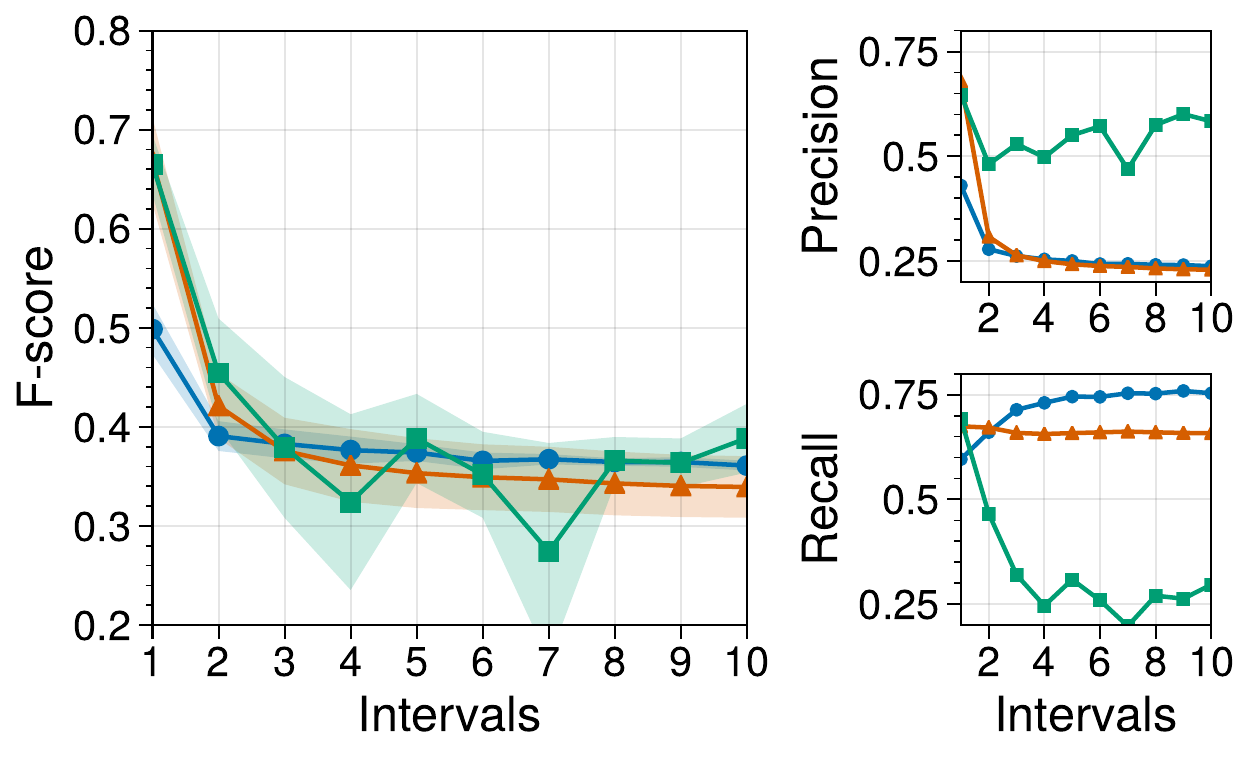}
 		\caption{Oracle challenger ($k = 2$)}
 	\end{subfigure}
 	\begin{subfigure}{\figurewidth\textwidth}
 	\includegraphics[scale=\figurescale]{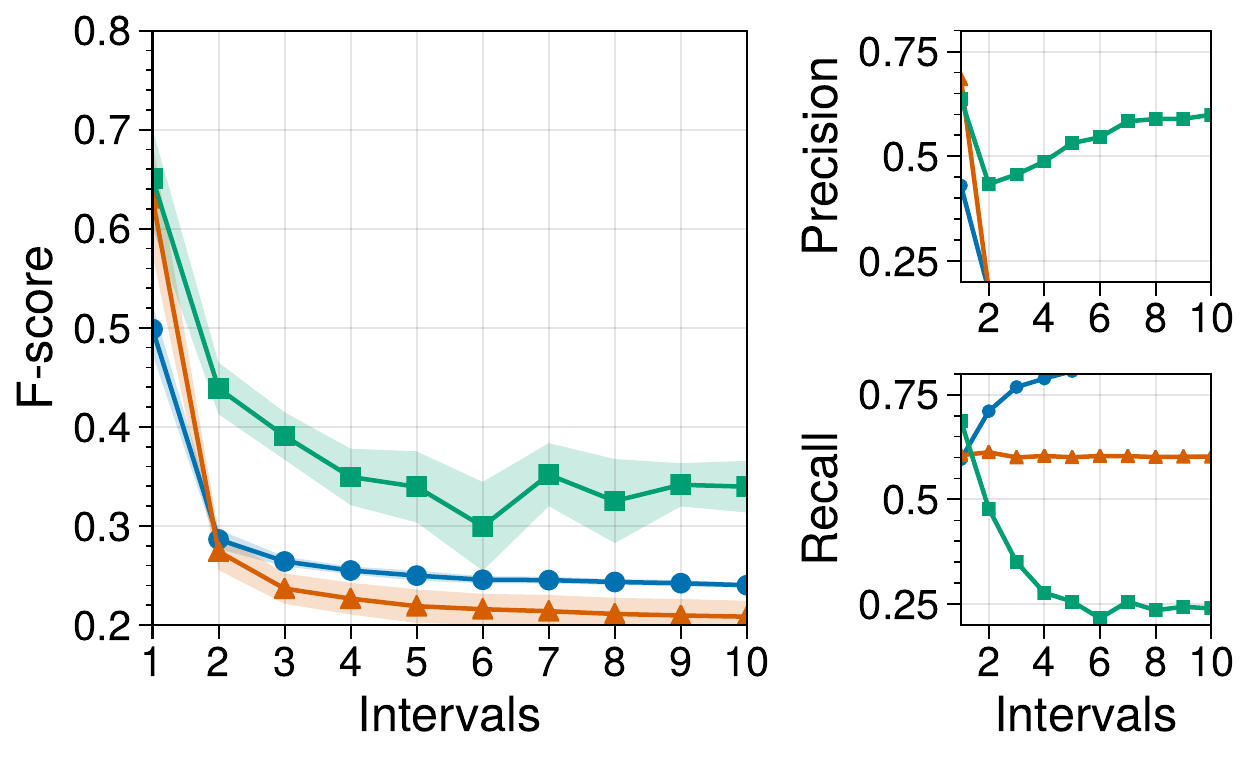}
 	\caption{Oracle challenger ($k = 5$)}
 	\end{subfigure}
 	\caption*{\textbf{(Black-box access.)} Adversaries (random, static and adaptive) in the presence of \textbf{oracle challenger} with $k=1,2,5$.} \label{fig:results_more2_oracle}
 \end{subfigure}

\vspace{0.2in}
\begin{subfigure}{\textwidth}
	\centering
	\begin{subfigure}{\figurewidth\textwidth}
		\includegraphics[scale=\figurescale]{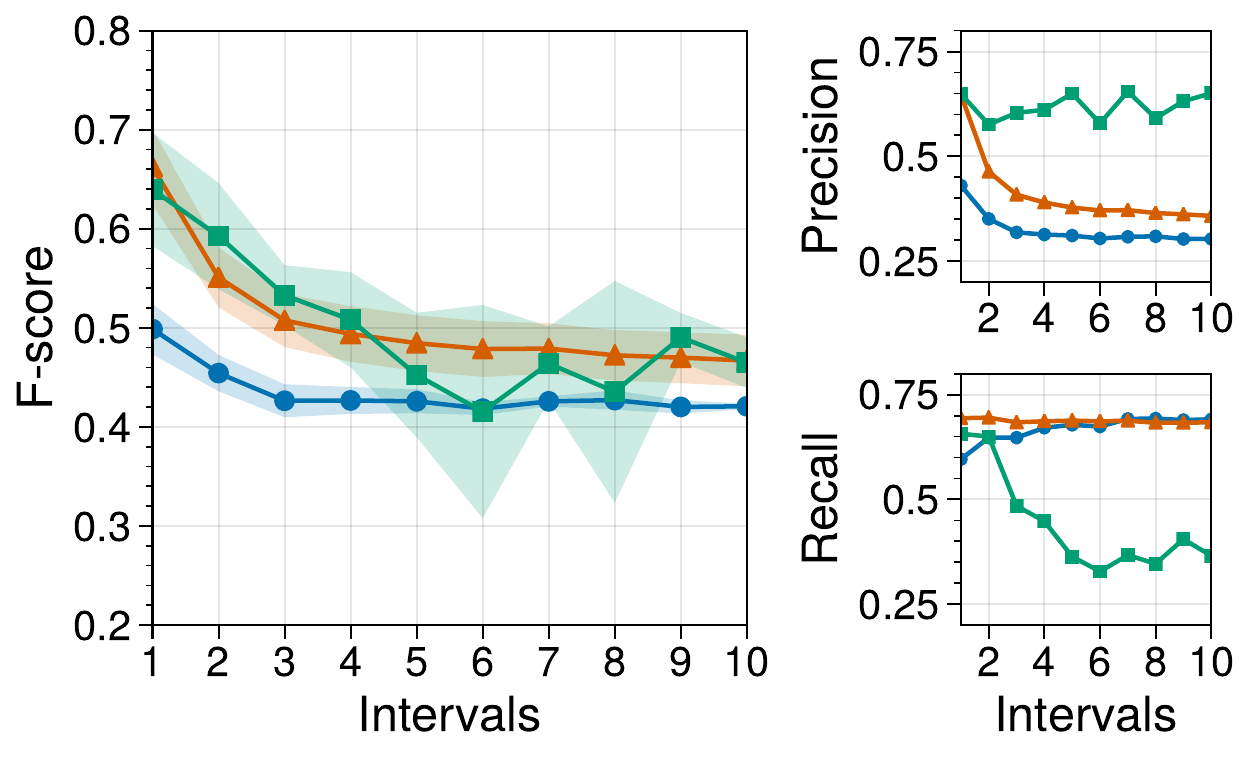}
		\caption{\system challenger ($k = 1$)}
	\end{subfigure}
	\begin{subfigure}{\figurewidth\textwidth}
		\includegraphics[scale=\figurescale]{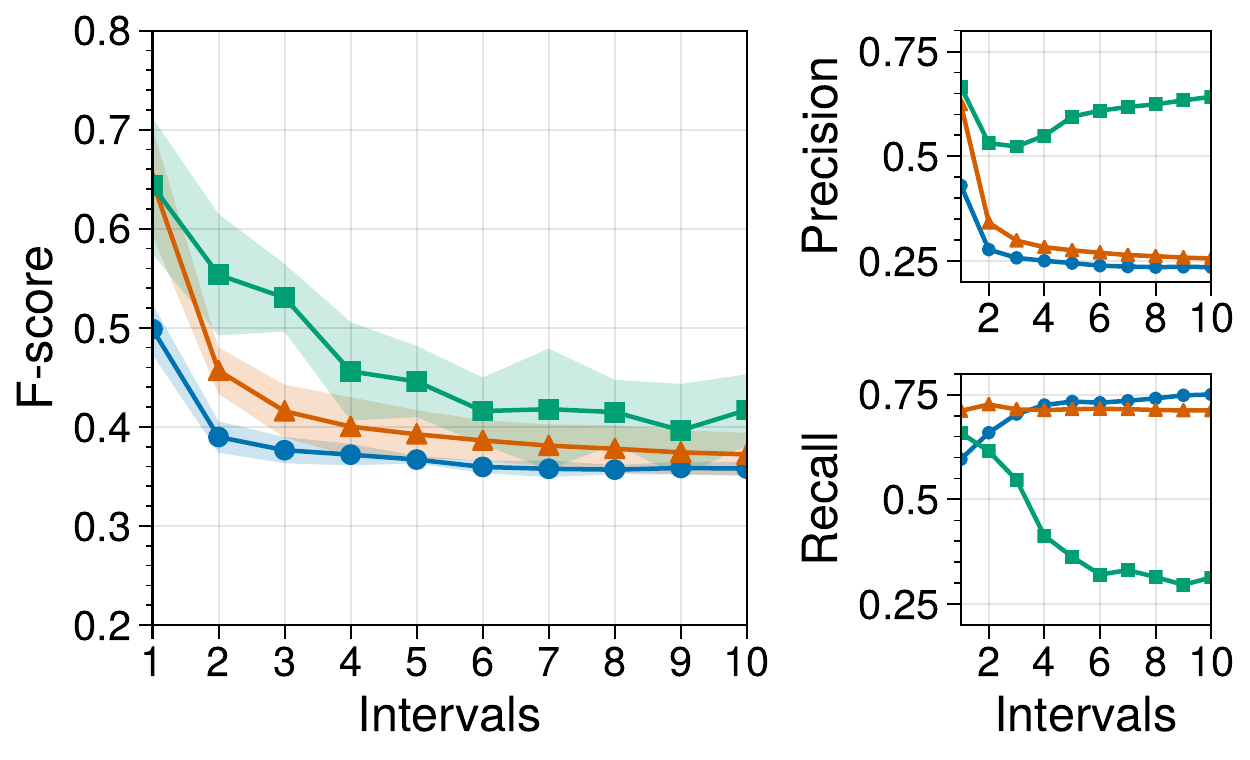}
		\caption{\system challenger ($k = 2$)}
	\end{subfigure}
	\begin{subfigure}{\figurewidth\textwidth}
		\includegraphics[scale=\figurescale]{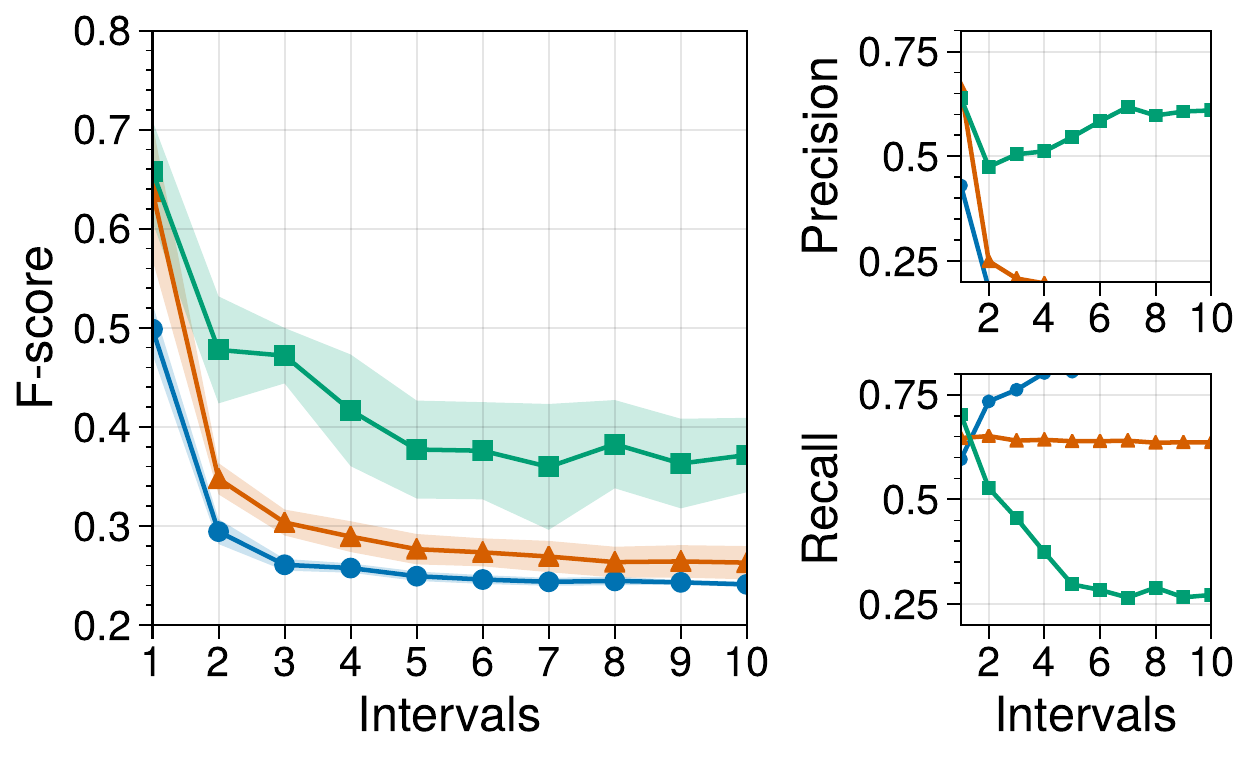}
		\caption{\system challenger ($k = 5$)}
	\end{subfigure}
	\caption*{\textbf{(Restricted black-box access.)} Adversaries (random, static and adaptive) in the presence of \textbf{\system challenger} with $k=1,2,5$. }
	\label{fig:results_more2_d2}
\end{subfigure}

    \caption{F-score (with 95\% confidence intervals), precision and recall for the three adversaries (random, static and adaptive) in the presence of different challengers corresponding to different accesses with $k=1,2,5$. Key observation: \system challenger fools the adversaries almost as well as the oracle challenger but with a \textit{restricted} black-box access.}
    \label{fig:entire_results_figure}
\end{figure*}


\paragraph{Detection of \damaging deletions in social media platforms is a serious concern}
We start by considering the case where no privacy-preserving deletion mechanism is in place (i.e., no \confuser to inject decoy deletions).  
In such a scenario, we compare the efficiency of different types of adversaries ten intervals shown in~\cref{fig:twitter_results_adversary}.

The {random adversary} labels the posts based on the prior distribution of the deleted tweets (around $42\%$ \damaging and $58\%$ non-\damaging every interval).
As expected, the adversary achieves a $42\%$ precision and $58\%$ recall resulting in an F-score of about $48\%$ in each interval.

As shown in~\cref{fig:twitter_results_adversary}, in the first interval, the static adversary achieves a $17$ percentage points (i.e., a 35\%) increase in its F-score compared to the random adversary, and remains almost constant over the rest of the intervals. On the other hand, the \adaptive adversary receives new training data every interval and trains its classifier continually, and hence is able to increase its F-score even further by about 10 percentage points (56\% increase compared to the random adversary) at the end of the $10$th interval. 

\textit{This shows that even normal users of social media platforms, not only celebrities and politicians, are vulnerable to the detection of their \damaging deletions}.
Furthermore, the adversaries can automate this attack on a large-scale with an insignificant amount of overhead (access to a small dataset of posts with the corresponding labels), highlighting the necessity for a much-needed privacy-preserving mechanism for the users' \damaging deletions in today's social platforms.

\paragraph{Injecting decoy deletions decreases the adversarial performance}
As explained in~\cref{sec:sys_threat_model,sec:model}, we consider three challengers corresponding to the three types of accesses to the adversary's model -- \textit{no access}, \textit{black-box access}, and \textit{restricted black-box access}. In the following, we compare the performance of the adversaries in the presence of the respective challengers against the adversaries' performance in the absence of any challenger.

\emph{No access:}
The top row of~\cref{fig:entire_results_figure} shows the performance of the three adversaries (random, \nonadaptive, and \adaptive) in the presence of the random challenger. 
We observe that although the F-score of both the static and the adaptive adversary decreases for all values of $k$, the reduction is not significant (only 7 percentage points for $k{=}1$ compared to the no-challenger case) In fact, both the adversaries still perform much better than the random adversary. This shows that protection of damaging deletions in the no-access scenario is possible but severely 
limited.

\emph{Black-box access:}
The middle row of \cref{fig:entire_results_figure} shows the performance of the adversaries in the presence of an oracle challenger. 
Not surprisingly, this approach is very effective at lowering the  (\nonadaptive and \adaptive) adversaries' F-scores (close to random for $k{=}1, 2$; i.e., 20 and 35 percentage point reduction in the case of $k{=}1$ for the \nonadaptive and \adaptive adversary respectively compared to the no-challenger case).

We also observe a major difference between the static and the adaptive adversaries in the presence of a competitive challenger. The static adversary retains the same recall performance (as in the no-challenger case) but loses drastically in precision, i.e., it classifies a large number of decoy posts as damaging. On the other hand, the adaptive adversary tries to \textit{adapt} to the presence of decoy posts and becomes highly conservative -- retains the same precision performance (as in the no-challenger case) but suffers heavily in the recall performance, i.e., it classifies a large number of damaging posts as non-damaging.

\textit{Restricted black-box access: }
The bottom row of \cref{fig:entire_results_figure} shows the performance of the adversaries in the presence of the \system challenger. 
The performance of the \system challenger is comparable to the oracle challenger. The adversaries' F-scores in the presence of the \system challenger is close to $45\%$ for the case of $k{=}1$ ($20$ and $30$ percentage point reduction for the \nonadaptive and \adaptive adversaries respectively compared to the no-challenger case). 
We also observe a precision-recall trade-off separating the static and the adaptive adversary (i.e., the static adversary loses in precision, whereas the adaptive adversary loses in recall) similar to the one described in the presence of an oracle challenger . 

Overall, we conclude that the \system challenger is able to successfully raise the bar for the adversaries in identifying damaging deletions \emph{without} requiring an unmonitored black-box access with infinite query budget.

\paragraph{The increase of decoy posts ($k$) results in lower adversarial performance with diminishing returns}
While examining each row of \cref{fig:entire_results_figure} individually, we see that the performance of the adversaries always decreases as $k$, the number of decoy deletions per damaging deletion, increases. However, we also observe that $k=1$ is enough to reduce the F-scores of the adversaries to 45\% (close to the random adversary). 

Since the goal of most social platforms is to retain as many posts as possible, it would \textit{not} be in the platform's best interests to use much larger values of $k$ or to delete the entire volunteered set.

\paragraph{Observation of damaging and decoy posts}
In~\cref{tab:sample_posts} in the Appendix, we show damaging tweets (as labeled by the AMT workers) and decoy tweets (chosen by the \system challenger from a set of non-deleted tweets).
We observe that even though the decoy tweets typically seem to have sensitive words, they do not possess content damaging to the owner.

%% file: design_rationale.tex
\section{Discussion} \label{sec:discussion}

\subsection{Adversarial Deception Tactics}
The adversary can use different techniques to sabotage the challenger. 
Here, we mention some prominent systems attacks
and their effects on the challenger.

\paragraph{Denial of Service attack}
One of such attacks could be a simple Denial of Service (DoS), where the attacker submits requests for many damaging deletions to consume all the volunteer posts. 
First, we remind that the volunteered posts are a renewable resource, not a finite resource, as the users create, volunteer and delete posts in each time interval. 
Regardless, a DOS attack is possible wherein the adversary can use up all volunteered posts collected up until this point.

A standard way to avoid such attacks is to limit the number of damaging deletions that can be protected for each user in one time interval (we assume that the adversary can have many \textit{adversarial users} to help with the DoS attack but is not allowed to use bots~\cite{stringhini2010detecting,chu2012detecting,dickerson2014using,wang2010detecting,bessi2016social,ferrara2016rise}). 
As is clear from~\cref{sec:whowins}, the challenger's defense is dependent on the distribution and number of volunteered posts. If there are more \textit{adversarial users} than volunteers, then the adversary can win the game.

We implemented the DoS attack as follows: in every interval, the adversary deletes as much as the standard deletions. 
We observed that the F-score did not change in this situation.

\noindent\paragraph{Volunteer Identification attack}
In a volunteer identification attack, the adversary deletes a bunch of posts and uses the process of doing so to identify individuals who volunteer posts to the challenger for deletion.
First, we note that in each time interval there is a large number of posts being deleted ($>100$ million tweets daily~\cite{lethePets2019}). 
Thus the posts deleted by the adversary (to try to identify volunteers) and the corresponding decoy deletions are mixed with other (damaging/non-damaging/decoy) deletions. 
In such a case, identifying the volunteers is equivalent to separating the decoy deletions from the damaging deletions; reducing to the original task.
Additionally, the challenger does not delete the decoy posts at the same time as the original damaging deletion but does so in batches spread out within the time interval.

Further, the volunteers can also have damaging deletions of their own. Even if an adversary is able to identify volunteers, the adversary still needs to figure out which of the volunteer's deletions are decoys. If the adversary ignores all posts from volunteers, then a simple protection for the users is to become a volunteer, which helps our cause.

\paragraph{Adversary disguising as volunteer}
In this attack, the adversary can take the role of a volunteer (or hire many volunteers) to offer posts to the challenger. Subsequently, the challenger may select the adversary's posts as decoys in the later intervals; however, these posts do not provide deletion privacy as the adversary will be able to discard these decoy posts easily. 
This effect can be mitigated with the help of more genuine volunteers and increasing the number of decoys per damaging deletion.
This points to a more fundamental problem with any crowdsourcing approach: if the number of adversarial volunteers is more than the number of genuine volunteers, the approach fails.

\subsection{Obtaining volunteered posts from users}
Volunteer posts are a significant component of our system. 
In \cref{sec:intro}, we describe the possibility of obtaining these volunteered posts via bulk deletions (i.e., whenever a user bulk-deletes, consider the posts as ``volunteered'' with a guarantee that they will be deleted within a fixed time period). 

However, other strategies could be more effective, for instance, one based on costs and rewards. Under such a strategy, each user seeking privacy for his/her damaging deletions is required to pay a cost for the service, whereas the users that volunteer their non-damaging posts to be deleted by the challenger (at any future point in time) are rewarded.
The costs and rewards can be monetary or can be in terms of the number of posts themselves (i.e., a user has to volunteer a certain number of her non-damaging posts to protect her damaging deletion). Nevertheless, in an ideal world, the volunteered set could also be obtained from {altruistic} users who offer their non-damaging posts for the protection of other users' deletions. 

We contacted the deletion services mentioned in~\cref{sec:intro} and shared our proposal Deceptive Deletions, for the privacy of users' damaging deletions. 
We got responses from some that provide services for the mass deletions on Twitter, Facebook, and Reddit. 
The response that we received has been positive.
They attest that, with Deceptive Deletions, an attacker that observes the deletion of users in large numbers will have a harder time figuring out which of the deleted posts contain sensitive material.

%% file: conclusion.tex
\section{Conclusion}

In this paper, we show the necessity for deletion privacy by presenting an attack where an adversary is able to increase its performance (F-score) in identifying damaging posts by 56\% compared to random guessing. Such an attack enables the system like Fallait Pas Supprimer to perform large-scale automated damaging deletion detection, and leaves users with ``damned if I do, damned if
I don't'' dilemma.

To overcome the attack, we introduce Deceptive Deletions (which we also denote as \confuser), a new deletion mechanism that selects a set of non-damaging posts (decoy posts) to be deleted along with the damaging ones to confuse the adversary in identifying the damaging posts. 
These conflicting goals create a minmax game between the adversary and the \confuser where we formally describe the \fullgame between the two parties. 
We further describe conditions for two extreme scenarios: one where the adversary always wins, and another where the \confuser always wins.
We also show practical effectiveness of \confuser over a real task on Twitter, where the bar is significantly raised against a strong \adaptive adversary in automatically detecting \damaging posts.
Specifically, we show that even when we consider only two \decoy posts for each damaging deletion the adversarial performance (F-score) drops to $65\%$, $42\%$ and $38\%$ where the \confuser has no-access, restricted black-box access and black-box access respectively. 
This performance indicates a significant improvement over the performance of the same adversary ($75\%$ F-score) when no privacy preserving deletion mechanism is in effect. As a result,
we significantly {\em raise the bar} for the adversary going after damaging deletions over the social platform.

Our work paves a new research path for the privacy preserving deletions which aim to protect against a practical, resourceful adversary. In addition, our \game can be adapted for current/future works in the domain of Private Information Retrieval~\cite{howe2009trackmenot,murugesan2009providing,domingo2009h,peddinti2010privacy} that have similar setting for injecting decoy queries to protect the users' privacy.

%% file: appendix.tex
\subsection{\fullgame vs Generative Adversarial Networks}\label{sec:gan}
Recall that in our setting, the task of the \confuser is to select posts from a pre-defined volunteered set $\cD^\offeredset$. An alternative approach is to use generative models~\cite{goodfellow2014generative,denton2015deep,ledig2017photo,ye2018yet,Radford2019} to generate fake texts ---see Zhang et al.~\cite{zhang2019generating} for a recent survey and Radford et al.~\cite{Radford2019} for the state-of-the-art--- enabling the \confuser to \textit{generate} decoy posts instead of selecting them from a pre-defined set. 
However, we note that such generative models might not be favorable or even effective in practical systems.

Let us consider the case of generating decoy posts on Twitter. 
Twitter posts are attached with a persistent non-anonymous user identities~\cite{correa-2015-anonymityShades}. 
Since, uploading fake posts from real user accounts raises serious ethical concerns, 
one should create multiple bot accounts that will upload machine-generated fake posts to be used as decoy posts (by deleting them later).
However, unfortunately, detection of bot accounts is a well studied problem~\cite{stringhini2010detecting,chu2012detecting,dickerson2014using,wang2010detecting,bessi2016social,ferrara2016rise}. 
Moreover, when an adversary detects a bot, any decoy post from that bot account will be similarly unmasked.
Therefore, in non-anonymous platforms like Twitter, selecting the decoy posts from the posts of actual users is arguably a more practical approach.

\subsection{\fullgame vs Adversarial Learning}\label{sec:adversarial_learning}
\noindent In traditional adversarial learning~\cite{dalvi2004adversarial} setting, there are two players: a classifier and an attacker. The classifier seeks to label the inputs $x$ (for instance, labeling emails as spam or not spam). Now, given a set of test inputs $\{x_i\}_{i=1}^N$, the attacker's goal is to \textit{modify} them such that the classifier will misclassify these examples (for example, in \cite{dalvi2004adversarial}, the attacker modifies spam emails to fool the spam-detector in labeling them as benign). The attacker is free to modify any example $x$ as long as humans would agree on its label (i.e., the attacker's modified email should still be considered as spam by humans). 

Our setting, however, is different in that we are \textit{not allowed to modify} the examples. Rather, the challenger wishes to attack the adversary's classifier by injecting \textit{hard-to-classify} examples into the adversary's dataset (i.e., the deletion set). A key constraint for the challenger is that it has to \textit{select the examples from a preexisting set}  of volunteered posts (i.e., $\cD^\offeredset$). This is because the challenger can only delete existing posts, and cannot generate fake posts (as we discuss in the next paragraph).

\subsection{Proofs}\label{app:proofs}

\begin{proposition*}[\cref{prop:confuser}.]
For any given volunteered set $\cD^\offeredset$ with $N$ non-deleted posts,
$$
\max_{\vphi} \tilde{V}(\vphi; \cD^\offeredset) = \max_{w_1, \ldots, w_{N}} V(w_1, \ldots, w_{N}; \cD^\offeredset)
$$
\end{proposition*}

\begin{proof}[Proof of \cref{prop:confuser}]
Let $S^*_1 = \max_{\vphi} \tilde{V}(\vphi; \cD^\offeredset)$
and $S^*_2 = \max_{w_1, \ldots, w_{N}} V(w_1, \ldots, w_{N}; \cD^\offeredset)$ be the optimum values for the respective objective functions. First, note that $S^*_1 \geq S^*_2$ because the optimal assignment for the discrete objective lies within the solution space of the continuous relaxation. Next, let $L_i = \log(1 - a(x_i; \vtheta_t))$, where $x_i$ is the $i$-th post in $\cD^\offeredset$ and let $\pi$ denote a sorting over them such that $L_{\pi_1} \geq \ldots \geq L_{\pi_{N}}$. 
Then, two cases arise -- (1) when the top $K$ elements are strictly greater than the rest, $L_{\pi_1} \geq \ldots \geq L_{\pi_{K}} > L_{\pi_{K+1}} \geq \ldots L_{\pi(N)}$, and (2) when there is atleast one element in the bottom $N-K$ elements that has the same value as one of the top $K$ elements, $L_{\pi_1} \geq \ldots \geq L_{\pi_{K}} = L_{\pi_{K+1}} \geq \ldots L_{\pi(N)}$. In the former case, the optimal solution is clearly to assign a weight of one to the top $K$ elements and zero to the rest. Any other assignment (even in the continuous solution space) is clearly suboptimal. In the latter case, although there are infinitely many optimal solutions in the continuous domain that distribute the weights differently among the equal elements, the value of the objective function is the same.
\end{proof}

\begin{proof}[Proof of Proposition~\ref{p:fullyoverlapping} (cont)]
First we show the kind of test distribution shift introduced by the \confuser. 
The challenger-injected distribution is given by the following hypothetical acceptance-rejection sampling algorithm:
\begin{enumerate}
    \item sample $x\sim p^\offeredset(x)$
    \item sample $u\sim Uniform(0,1)$  independently of $x$
    \item while  $u>p^\damset(x)/(M p^\offeredset(x))$, reject  $x$ and GOTO 1, for some constant $M$.
    \item Accept (output)  $x$  as a sample from $p^\damset(x)$ but with label $y=0$, as the sample came from $p^\offeredset(x)$.
    \item While number of samples less than $k |\cD^\damset|$, GOTO 1
\end{enumerate}
Next we prove that the above rejection sampling algorithm produces samples with distribution $p^\damset(x)$ from examples from decoy examples that have distribution $p^\offeredset(x)$.
Let $X'$ be a sample from the algorithm described above and $X \sim p^\damset(x)$, then $$ p(X' = x) = p(X = x| \text{Accept}) =  \frac{p(X = x, \text{Accept})}{p(\text{Accept})} = p^\damset(x)
$$
because 
\begin{align*}
\frac{P(X = x, \text{Accept})}{P(\text{Accept})} &= \frac{P(\text{Accept}|X = x) p(X = x)}{P(\text{Accept})} \\ &= \frac{\frac{p^\damset(x)}{M p^\offeredset(x)} p^\offeredset(x)}{P(\text{Accept})} \\ &=  \frac{\frac{p^\damset(x)}{M}}{P(\text{Accept})}= p^\damset(x)
\end{align*}
as
\begin{align*}
P(\text{Accept}) &=
\int P(\text{Accept}|X = x)p(X = x) dx\\
& = \int \frac{p(x)}{M q(x)} q(x) dx \\
&= \frac{1}{M} \int p(x) dx = \frac{1}{M} 
\end{align*}

The above ideal accept-reject sampling procedure can be reproduced via noise contrastive estimation~\cite{gutmann2010noise}, which is  method that can generate data from a known distribution without the need to know $p^\damset(x)/(M p^\offeredset(x))$ in advance. 
A variant of the same statistical principle is used today in generative models using Generative Adversarial Networks~\cite{goodfellow2014generative}, which uses a minimax game similar to our procedure.
Because we train the \confuser to mimic the classifier of the adversary, it is easy to construct such rejection sampling method, such that there are in average $k$ decoy examples for every damaging example in the original data.
\end{proof}

%% file: amt_details.tex

\begin{table*}[!tb]
	\centering
	
	\caption{Sample tweet text extracts from the damaging, decoy, and non-damaging datasets. The real user accounts within the tweets have been replaced with @UserAccount. Some letters in the offensive keywords have been replaced by *. \label{tab:sample_posts}}
	\small
	\resizebox{\textwidth}{!}{
		\begin{tabular}{l r}
			\toprule
			\textbf{Tweets' text extract} & \textbf{Tweet Type} \\ 
			\midrule
			``\#GrowingUpInTexas Seeing a black person pass by ya front yard and telling your son to pass you the shotgun so you can play shoot em ups'' & damaging\\
			``@UserAccount its gods way of punishing you for your sins. fag**t.'' & damaging\\
			``I think I might have the biggest douche for a boss hands down breaking point'' & damaging\\
			``I don t wanna believe all the women in the auto department at walmart are lesbians    Someone prove me wrong    Cuz im seeing it'' & damaging\\
			``Show up to work on meth once and your nickname is  Tweaker  for the rest of your life '' & damaging\\
			``I cook that pot them junkies treat me like Obama'' & damaging\\
			``gotta love watching two gay men having sex next to my homophobic parents'' & damaging\\
			\hline
			\hline
			``Listening to this deuchbag behind me at Chipotle diss every girl who comes in   hot body  but she has no face    news check  you re fugly'' & decoy\\
			``I grab a beer from the fridge  put on my Bob Marley record  crank that f**ker up  and light up a fat one   my professor is the shit'' & decoy\\
			``Kids having kids  That sh*t is f**kin crazy to me  I d rather be that cool ass uncle that buys the booze aaayyye'' & decoy  \\
			``If you know the boy is in a relationship  and you continue to hook up with him  yes sweetheart sorry to break it to you  you re a wh*re'' & decoy\\
			``This guy smacked his girlfriends ass in public  That's disrespectful'' & decoy\\
			``I don't understand why people say that watermelon and fried chicken is for black people    I love that shit to  Dafuq'' & decoy\\
			``You think communism looks good on paper  Get the hell out of America   USA  damncommie'' & decoy\\
		    \hline
		    \hline
             ``y'all I just watched ``love, simon'' for the first time and let me just say that the ugly tears are so f**kin real omfg'' & non-damaging\\
		    ``I want to eat to rid my emotions but I don't want the calories ya feel me'' & non-damaging\\
		    ``The middle is not the end, but a process you must grow through to get to your new beginning.'' &non-damaging\\
		    ``@UserAccount @UserAccount Did you guys win it or did you burgle it from a classmate’s house again?'' &non-damaging\\
		    ``Love is not about turning human relationships into billions of isolated couples.'' &non-damaging\\
            ``Gonna do some coffee shop hopping tomorrow if anyone wants to join'' & non-damaging\\
            ``I’m pretty sure one of my professors has me mistaken for another black woman in my class.'' & non-damaging\\
            ``Anyways, it's 2 am and the Full House theme song is playing in my head on repeat so if you wanna beat me to death do it now please'' & non-damaging\\
            
			\bottomrule
		\end{tabular}
	}
\end{table*}

\subsection{Examples of Damaging and Decoy Posts}
\cref{tab:sample_posts} presents the damaging tweets (deleted tweets that were labeled as damaging by the AMT workers), the decoy tweets (chosen by the \system challenger) and the non-damaging tweets.